\documentclass[aps,prl,twocolumn]{revtex4-1}
\usepackage{epsfig,amssymb,amsmath,mathrsfs,amsthm,graphicx,float,color}
\usepackage[active]{srcltx}
\usepackage{subfigure}
 \def\map#1{\mathcal #1}
\def\d{\operatorname{d}}\def\<{\langle}\def\>{\rangle}

\def\Tr{\operatorname{Tr}}\def\:{\hbox{\bf
    :}}

\def\grp#1{\mathsf{#1}}
\def\st#1{\mathbf{#1}}

\def\Span{\mathsf{Span}}
\def\Supp{\mathsf{Supp}}

\def\spc#1{\mathcal{#1}}

\def\set#1{\mathsf{#1}}
\def\Fix{{\mathsf{Fix}}}
\def\Proof{{\bf Proof. }}  

\def\rank{\mathsf{rank}}
\newtheorem{theo}{{Theorem}}
\newtheorem{lem}{{Lemma}}
\newtheorem{prop}{{Proposition}}
\newtheorem{cor}{{Corollary}}

\begin{document}
\title{Efficient Quantum  Compression for Ensembles of Identically Prepared Mixed States }
\author{Yuxiang Yang,  Giulio Chiribella,  and  Daniel Ebler} 
\affiliation{Department of Computer Science, The University of Hong Kong, Pokfulam Road, Hong Kong}
\begin{abstract}
We present one-shot compression protocols that  optimally encode   ensembles   of   $N$ identically prepared mixed states into $O(\log N)$  qubits.  In contrast to the case of pure-state ensembles, we find that     the number of  encoding qubits     drops down discontinuously  as soon as a nonzero error is tolerated and  the spectrum  of the states is known with sufficient precision.    For qubit ensembles, this feature leads to a 25\% saving of memory space.  
Our compression protocols    can be implemented efficiently on a quantum computer.

\end{abstract}
\maketitle

Storing data into the smallest possible space is  of crucial importance in present-day digital technology,  especially when dealing with large amounts of information and with limited memory space \cite{bigdata}.    
The need for saving space is even more pressing in the quantum domain, where storing data is an expensive task that requires sophisticated error correction techniques \cite{Memory1,Memory2,Memory3}.   

For quantum data,   Schumacher's compression  \cite{Schumacher} and its extensions \cite{JozsaSchumacher,LoBound,Horodecki,
Jozsa,BennetHarrowLloyd} provide  optimal ways  to store  information  in the asymptotic limit of many identical and independent uses of the same source.    
 However,     in  many situations  there may be  correlations  from one use of the source to the next.   In such situations,  it is convenient to regard    $N$  uses of the original  source as   a single use of a new source,  which  emits messages of length $N$.     This scenario is an instance of  one-shot quantum data compression \cite{datta}.  
   An important example of one-shot  compression is when the states  emitted at $N$ subsequent moments of  time are perfectly correlated, resulting in codewords of the form $\rho_x^{\otimes N} $ for some density matrix $\rho_x$ and some  random parameter $x$.  This situation arises when the original source is an uncharacterized preparation device, which generates the  same quantum state at every use.      For    quantum bits (qubits),  Plesch and Bu\v zek  \cite{SchurTransform2}  observed that every  ensemble of identically prepared pure states can be stored without any error into  $\log (N  +1)$ qubits,  thus allowing for  an exponential saving of  memory space.  
  Recently, Rozema \emph{et al} \cite{SchurCompression} brought this idea into the realm of experiment,   demonstrating   a prototype of one-shot compression  in  a photonic setup.  
  
The possibility of implementing one-shot compression in the lab opens  new questions that require one to go beyond the ideal case of pure states and no errors. 
 First, due to the presence of noise, real-life implementations typically involve  mixed states---think, ~e.~g., ~of  quantum information processing with NMR  \cite{NMR}, where the standard is to have thermal states at a given temperature, or, more generally, of mixed-state quantum computing  \cite{aharonov-kitaev,1bit,shor-jordan,datta2008quantum,lanyon-barbieri}.    For mixed states, the basic principle of pure-state  compression does not work:  in the qubit case, for example, projecting the quantum state into  the smallest subspace containing the code words  does not lead to any compression if the   states  $\rho_x^{\otimes N}$ are mixed, because in that case  the smallest subspace is the whole Hilbert space.       As a result,  it is natural to search for compression protocols that work for mixed states and to ask which protocols achieve  the best compression performance. 
  An even more  important question is how the number of qubits needed to store data  depends on the  errors  in the decoding.     Tolerating a nonzero error is natural in real-life implementations, which typically suffer from noise and imperfections.     

In this Letter we  answer the above questions, proposing  compression protocols  for ensembles of identically prepared mixed states.   We first analyze the zero-error scenario, showing  that the storage of   $N$ mixed qubits with known purity and unknown Bloch vector requires a quantum memory of at least  $2 \log N$ qubits. 
The size of the required memory is twice that of the required memory for pure states, but it is still exponentially smaller that the initial data size.    
The maximum compression is achieved by a protocol that does not require knowledge of the purity.
We then investigate the more realistic case of protocols with an error tolerance.    When the purity  is known with sufficient precision,  we find out that tolerating an error, no matter how small, allows one to encode the initial data into only  $3/2  \,  \log N$  qubits, plus a small correction independent of $N$.   Remarkably, the  discontinuity   in the error parameter takes place as soon as the prior knowledge of the  purity  is more precise than the  knowledge that could be gained by measuring the  $N$ input qubits.      The existence of a discontinuity is a striking deviation from the pure-state case, for which  we prove that  there is no significant advantage in introducing an error tolerance.     
Furthermore, we show that our compression protocol  can be implemented efficiently and that the compression rate is optimal under the requirements that the encoding be rotationally covariant and the decoding preserve   the magnitude of the total angular momentum. These assumptions are relevant in physical situations where the mixed states  are used as indicators of spatial directions   \cite{demkowicz,bagan} and the decoding operations are limited by conservation laws \cite{reference, marvian-spekkens1,marvian-spekkens2,spekkens-marvian-WAY, ahmadi,marvian-spekkensNat}.    
     All our results can be generalized to quantum systems of arbitrary finite dimension, where  we quantify how the presence of degeneracy in the spectrum affects the compression rates.

Let us start from the   qubit  case, assuming $N$ to be even for the sake of concreteness.  
We denote by $\map{E}: \spc{H}^{\otimes N}\to\spc{H}_{\rm enc}$   ($\map{D}: \spc{H}_{\rm enc}\to \spc{H}^{\otimes N})$ the encoding  (decoding) channel, where $\spc H$ is the Hilbert space of a single qubit and $\spc H_{\rm enc}$ is the Hilbert space of the encoding system.   For an ensemble  of identically prepared qubit states $\{ \rho^{\otimes N}_x \,  ,  p_x  \}$ the average error of the compression protocol is 
\begin{align}\label{error}
e_N= \sum_{x} \, p_x  \frac{ \left\|\rho^{\otimes N}_x-\map{D}\circ\map{E}\left(\rho_x^{\otimes N}\right)\right\|} 2 \, ,
\end{align}
$\|  A  \|$ denoting the trace norm. 
 We  consider ensembles where all the states  $\rho_x$   have the same  purity,  which is assumed to be perfectly known (this assumption will be lifted later).  Let us write $\rho_x$ as   $\rho_{\st n}=p \, |  \st n  \>\<  \st n|  \,  + (1-p)  \,  | - \st n \>\< -\st n |$, where   $|\st n\>$ denotes  the two-dimensional pure state with Bloch vector    $\st n  =  (n_x,n_y,n_z)$ and $p  \ge 1/2$ is the maximum eigenvalue.   
  We focus on mixed states $(p\not  = 1)$, excluding the trivial case $p=1/2$, in which the ensemble consists of just one state. For $p\not \in  \{ 1,1/2\}$,   we call the ensemble  $\{\rho_{\st n}^{\otimes N}\, , p_\st n\}$  complete if  the probability distribution  $p_{\st n}  $ is dense in the unit sphere.  
 The typical example is an ensemble of mixed states with known purity and completely unknown Bloch vector.     For every complete ensemble we demonstrate a  sharp contrast between  two types of compression: (i)  zero-error compression, wherein the decoded state is equal to the initial state, and (ii) approximate compression, wherein small errors are tolerated.  
  In the zero-error case we have the following 
\begin{theo}
\label{thm:zeroerror}
The minimum number of  logical qubits needed to compress a complete   $N$-qubit ensemble  is  $ \lceil    2  \log  (N+2)-2\rceil $. Every compression protocol that has zero error on a complete ensemble must have zero error on every ensemble of identically prepared mixed states and on every ensemble of permutationally invariant N-qubit states.
\end{theo}
Intuitively, the reason for the exponential reduction of the number of qubits is that  the states in the ensemble are invariant under permutations and, therefore, they do not carry all the  information that could be encoded into $N$ qubits.  This observation was anticipated by Blume-Kohout \emph{et al}  in the context of state discrimination and tomography \cite{estimation}.  The key point  of Theorem   \ref{thm:zeroerror} is the optimality proof, which establishes that if a mixed-state  ensemble  is complete, then compressing it is as hard as compressing any arbitrary ensemble of permutationally invariant states \cite{supplemental}.      

In preparation of our analysis of approximate compression, it is instructive to look into an optimal  protocol achieving zero-error compression.  The starting point  is the Schur-Weyl duality \cite{fultonharris}, 
 stating  that there exists a basis   in which the $N$-fold tensor action of the group $\grp{GL}(2)$  and the natural action of the permutation group $S_N$  are both block diagonal.  
In this basis, the Hilbert space of the $N$ qubits can be decomposed as 
\begin{align}
\spc{H}^{\otimes N}\simeq\bigoplus_{j=0 }^{N/2}\left(\spc{R}_j\otimes\spc{M}_{j}  \right)   , 
\end{align}
 where $j$ is the quantum number of the total angular momentum,  $\spc{R}_j$ is a representation space, in which the group $\grp{GL}(2)$ acts irreducibly, and $\spc{M}_{j}$ is a multiplicity space, in which the group acts trivially.  
 Now, since the state $\rho_{\st n}^{\otimes N}$ is invariant under  permutations of the $N$ qubits, one has  
\begin{align}\label{statedecomp}
 \rho_{\st n}^{\otimes N}  =\bigoplus_{j=0}^{N/2}  \,  q_{j,N}  \, \left(\rho_{\st n,j}\otimes \frac{I_{m_j }}{m_j}\right),
\end{align}
where  $q_{j,N}$  is a suitable probability distribution in $j$,      $\rho_{\st n,j}$ is a quantum state on $\spc{R}_j$, $I_{m_j} $ is the identity on  $\spc{M}_{j} $,  and $m_j$ is the dimension of $\spc{M}_j$.
From Eq. (\ref{statedecomp}) it is obvious  that all information about the input state lies  in  the representation spaces.  
  Hence, $\rho_{\st n}^{\otimes N}$ can be encoded faithfully into the state $  \map E  \left (\rho_{\st n}^{\otimes N}  \right)   =  \bigoplus_{j}  q_{j,N}\,  \rho_{\st n, j}$.   Such state has an exponentially smaller support, contained in the space $\spc{H}_{N}:=\bigoplus_{j=0 }^{N/2}\spc{R}_j$, whose dimension is  
$\dim\spc{H}_{N}=\left( N/2+1\right)^2$.
Hence, the initial state can be encoded into  $\lceil\log \dim\spc{H}_{N}\rceil$ qubits---the amount declared in  Theorem \ref{thm:zeroerror}.  A perfect decoding is achieved by  the  channel 
  \begin{align}\label{channelD}
\map{D}(\rho) :  =    \bigoplus_{j   }  \,     \left(     P_j  \,\rho \, P_j  \otimes \frac{ I_{m_j}}{m_j}\right)  \, ,
\end{align} 
where $P_j$ is the projector on the representation space $\spc R_j$.

Considering that qubits  are a costly resource, it is worth pointing out a slight modification of the above protocol, which uses approximately $\log N$ qubits and $\log N$  classical bits.  The modified protocol consists in (i)  measuring the value of $j$, thus projecting $N$ qubits into the state $\rho_{\st n,j} \otimes I_{m_j}/m_j$, (ii)   discarding the multiplicity part,   (iii)  encoding the state $\rho_{\st n,j}$ into   $\lceil \log (N+1)\rceil$ qubits,  and (iv) transmitting the encoded state to the receiver, along with a classical message specifying the value of $j$. Knowing the value of $j$, the receiver can append  an additional system in the state $I_{m_{j}}/{m_j}$  and  embed the  state  $\rho_{\st n,  j}\otimes I_{m_j}/m_j$  into the right  subspace. 

Let us consider now  the more realistic case of approximate compression.  Here,  the number of encoding  qubits  drops down discontinuously.

\begin{theo}\label{thm:faithful}
For every allowed error rate $\epsilon  >0  $ and for every complete qubit ensemble,  
there exists a number $N_0>0$ such that for any $N\ge N_0$ the  ensemble   can be encoded into   $3/2 \log N+\log[4(2p-1)\sqrt{\ln(2/\epsilon)}]$ qubits  with error smaller than $\epsilon$. 
\end{theo}
 The idea is to work out the explicit form of the probability distribution $q_{j, N}$ in Eq. (\ref{statedecomp}), given by 
\begin{align}\nonumber
q_{j,N}= \frac{2j+1}{2j_0}   &\left[   B\left(N+1,p,\frac N2 +  j+1 \right)\right.\nonumber \\
&\left. ~ -B\left(N+1,p,\frac N2 - j \right)\right] \label{dist}
\end{align}
where $B(n,p,k)$ is the binomial distribution  with $n$ trials and with probability $p$, and $j_0 = (p-1/2)(N+1)$.  
For  large $N$, the distribution $q_{j,N}$ is approximately the product of a linear function with the normal distribution of variance   $(N+1)p(1-p)$ centered around $j_0$.  In order to  compress, we get rid of the tails:  for every $\epsilon  >0  $, we  select a set    $\set S_\epsilon : =  \left\{      j_0 -  \lfloor   \sqrt{\ln(2/\epsilon)N}  \rfloor  ,  \dots,     j_0 +  \lfloor   \sqrt{\ln(2/\epsilon)N}   \rfloor  \right\}$  and we   compress the state   $\rho_{\st n}^{\otimes N}$   into the encoding space
$\spc{H}_{\rm enc}=\bigoplus_{j\in \set S_\epsilon }\spc{R}_j$, by applying the quantum channel 
\begin{align}\label{channelE}
\map{E}(\rho):=  \bigoplus_{j \in  \set S_\epsilon} \,   \Tr_{  \spc M_j}   \, \left [   \,    \Pi_j \,     \rho \,  \Pi_j   \,  \right]  +    \sum_{j\not \in\set S_\epsilon  }  \,\Tr\left[      \Pi_j  \, \rho\right]  \, \rho_0\, ,
\end{align}
where   $\Pi_j$ is the projector on $\spc R_j \otimes \spc M_j$, $\Tr_{\spc M_j}$ is the partial trace over $  \spc M_j$,   and $\rho_0$ is a fixed state with support inside $\spc{H}_{\rm enc}$.    
  The encoding space has   dimension
 \begin{align*}
\dim\spc{H}_{\rm enc}   &  =  \sum_{j\in\set S_\epsilon} \,   (2j +1) \le       \, (2j_0  +1) \left(2 \sqrt{  N \ln\frac{2}{\epsilon}}+ 1\right)   \, ,  
\end{align*}
growing as $N^{3/2}$.  
The initial state can  be recovered,  up to  error $\epsilon$, by  a suitable decoding channel  \cite{supplemental}.

 Theorem \ref{thm:faithful} guarantees   that  $N$ identical copies of a mixed state with known purity  can be stored faithfully to $\epsilon$  into   $3/2\log N$ qubits,  plus an overhead   that is doubly logarithmic  in $1/\epsilon$.  This result is good news for  future  implementations,  because the overhead  grows slowly with the required accuracy.   For example, when $p=0.6$, 
  $N=   20$ identically prepared qubits  with Bloch vectors pointing in arbitrary direction can be compressed into 8 qubits with an error smaller than $1\%$. 
 In addition to the fully quantum version of the protocol, one can construct a hybrid version where the initial state is stored partly into qubits and partly into classical bits, as  discussed  in the zero-error case.    In the hybrid version, the discontinuity between zero-error and approximate compression  pertains to  the number of classical bits needed to communicate the value of $j$, which decreases from $\log N$ to $1/2  \log N$ as soon as a nonzero error is tolerated.

Our result highlights a  radical difference between  mixed and pure states: for mixed states, every finite  error tolerance  $\epsilon  >  0$  allows one to reduce the size of the compression space from the original $2  \,  \log N$    qubits   to  $3/2 \,  \log N$  qubits.     Such a  discontinuity  does not take place for pure states:  for pure states with completely unknown Bloch vector,  every compression protocol    
 with tolerance $\epsilon$  requires at least   $(1-  2  \epsilon)\,  \log  N$  qubits \cite{supplemental}.

It is worth commenting  on the importance of knowing the purity.  Our approximate protocol requires  the purity to be perfectly known, so that one can  encode only  the subspaces where the quantum number $j$ is in a strip around  the most likely value.    If the purity is   only partially known,  the protocol can be adapted  by broadening the size of the strip, i.\,e., by changing the set $\set{S}_\epsilon$.   Specifically, suppose that  the eigenvalues of $\rho_{\st n}$ are known up to an error $\Delta p =  O(  N^{-\gamma})$, with $\gamma  \ge 1/2$. In this case,  the number of encoding qubits can be reduced to $3/2\,  \log N+g(\epsilon,\gamma)$ where $g$ is a function depending on $\epsilon$ and $\gamma$, but not on $N$.  Hence, the discontinuity between zero-error and approximate compression persists.    However, the situation is different if the eigenvalues are known with less precision:   if the  error in the specification of the eigenvalues scales as   $ N^{-\gamma}$ with $\gamma < 1/2$, then the number of encoding qubits becomes $(2-\gamma)\,  \log  N$.   Quite intriguingly, the separation between the two regimes takes place exactly when the knowledge of the eigenvalues becomes more precise than the knowledge that could be extracted through spectrum estimation  \cite{KeylWerner}.    Note that our protocol  can be combined for free with spectrum estimation, which only requires measuring the value of $j$. 
  However, the \emph{a posteriori} knowledge of the measurement outcome cannot replace  the \emph{a priori} knowledge of the spectrum:   indeed,  finding the outcome  $j$ leads to  estimating the  maximum eigenvalue  as $\hat p =1/2  +  j/(N+1)$   \cite{KeylWerner} and then to encoding the state $\rho_{\st n,j}$ into $\lceil \log(2j+1)\rceil$ qubits. In order to decode, the receiver needs a classical message communicating  the value of $j$, which requires $\lceil \log (N/2+1)  \rceil$ bits in the one-shot scenario. This  leads to the same resource scaling  as in the zero-error case, i.~e., approximately $\log N$ qubits to send the encoded state and $\log N$ bits to communicate $j$.

The  protocol of   Theorem \ref{thm:faithful} is optimal within the physically relevant class of protocols constrained  by covariance under rotations and by the preservation of the magnitude of the angular momentum.     
More precisely,   we have the following  \cite{supplemental}.
\begin{theo}\label{thmopt}     Every compression protocol that encodes a complete  $N$-qubit ensemble  into   $(  3/2  -  \delta)  \,  \log N$ qubits  with  covariant encoding and a decoding that  preserves the magnitude of the total angular momentum  will necessarily have error $e \ge1/2$ in the asymptotic limit. 
\end{theo}

\begin{figure}[t!]
      \includegraphics[width=0.45\textwidth]{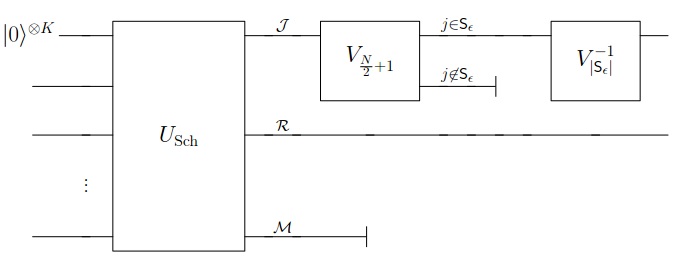}\caption{{\bf A quantum circuit for encoding.}   The Schur transform turns the initial $N$ qubits together with $K=  O(\log N)$ ancillary qubits into three registers: the index register $\spc J$, the representation register $\spc{R}$, and the multiplicity register $\spc{M}$. The multiplicity register is  discarded.   The index register is encoded into $N/2+1$ qubits by the position embedding $V_{N/2+1}$.  The qubits in positions outside $\set S_\epsilon$ are discarded and the remaining qubits are reencoded into $\lceil  \log  |\set S_\epsilon|  \rceil$ qubits. }
       \label{fig:encoding}
\end{figure}
      
\begin{figure}[t!]      
       \includegraphics[width=0.45\textwidth]{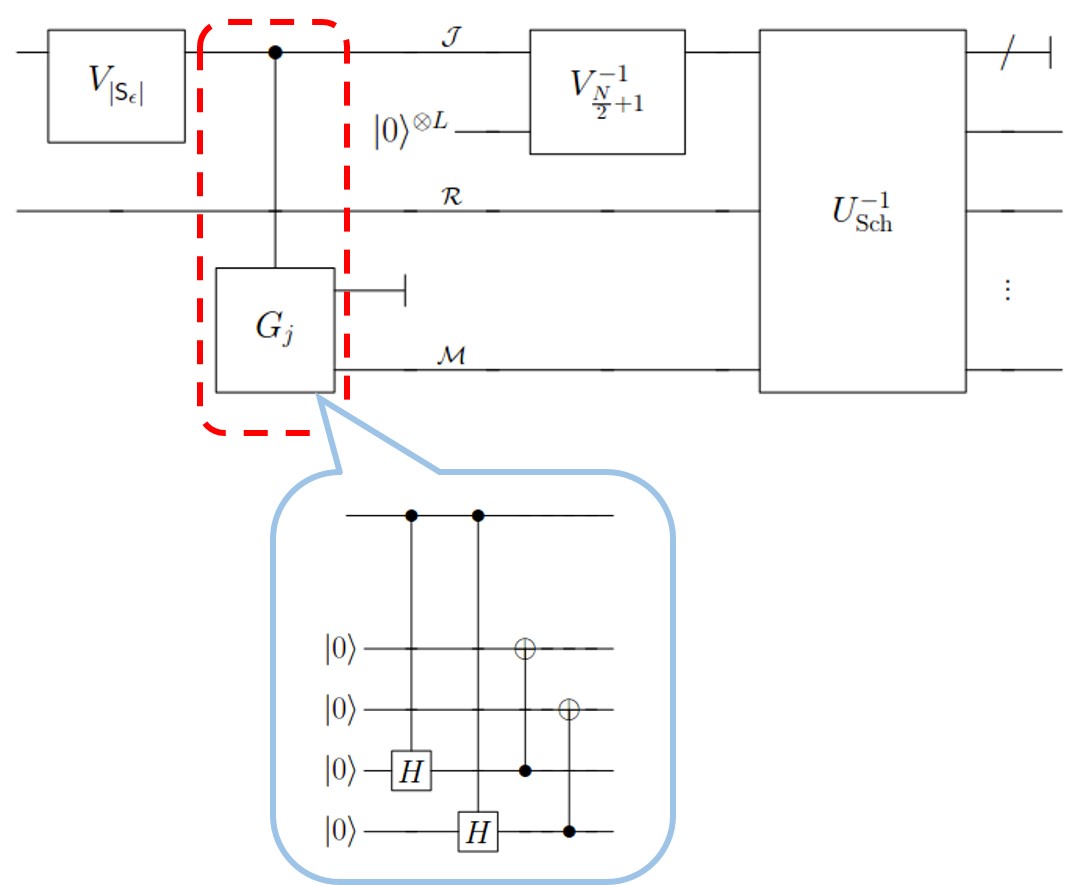}
       \caption{{\bf A quantum circuit for decoding.}    The first operation   is the position embedding   $V_{|  \set S_\epsilon|}$, which produces  $|\set S_\epsilon|$ output qubits.   The $j$th of these qubits   controls the generation of a maximally mixed state of rank $m_j$ (achieved by the controlled operation  $G_j$, represented explicitly in the blue inset for $m_j=4$).   The third step is the initialization of   $L=  N/2+1  -   |\set S_\epsilon|$  qubits which are put in positions corresponding to values of $j$ outside $\set S_\epsilon$.   After a total of  $N/2  +1$ qubits are in place, the inverse of the position embedding is performed, followed by the inverse of the Schur transform. The  output of the circuit is a state on $N$ qubits and $  K  =  O(\log N)$ ancillas, which are finally discarded. }\label{fig:decoding}   
\end{figure}

Let us now discuss the complexity of the compression protocol.    To operate on the input state we use the Schur transform \cite{harrowthesis,SchurTransform,SchurTransform2}, which transforms the initial $N$ qubits  together with $O(\log N)$ ancillary qubits into three registers:  (i) the index register, where the value of $j$ is stored into the state of $\log  (N/2+1)$ qubits,  (ii)  the representation register,  which uses $\log (N+1)$ qubits to encode the representation spaces, and (iii) the multiplicity register, where the multiplicity spaces are encoded into  $O(N)$ qubits (see  Fig. \ref{fig:encoding}).          
 Since the implementation of the Schur transform in a quantum circuit  is approximate, we focus on  approximate compression, so that  the Schur transform error can be absorbed into the compression error.     Let us analyze first the encoding.  The first step is the approximate Schur transform,  whose complexity is   ${\rm poly}( N, \log 1/\epsilon')$,  $\epsilon'$ being the approximation error   \cite{harrowthesis,SchurTransform}. 
We  set $\epsilon'$  to be vanishing exponentially  in $N$,   resulting in a  complexity  ${\rm poly}  (N)$ for the implementation of the Schur transform.  
 After the Schur transform has been performed, the encoding  circuit   embeds the  index register  into an exponentially larger register of $N/2+1$ qubits, transforming the state $|j\>$ into the state where the $j$th qubit is set  to $|1\>$ and the rest of the qubits are set to $|0\>$      \cite{SchurTransform2}.      
We  refer to this  transformation as position embedding and denote it by $  V_{D}$, where $D$ is the dimension of the register that is being embedded (in this case $D=N/2+1$).   The point of position embedding is to physically  encode the value of $j$ in a form that makes it easy  to check whether or not $j$ belongs to the set $S_\epsilon$. In fact, such a check can be equivalently implemented on a classical computer.        After this step, the circuit discards the qubits in positions outside the set $S_\epsilon$ and transforms the remaining qubits into $\log |S_\epsilon|$ qubits, by applying  $V^{-1}_{  |S_{\epsilon}|}$.  
    Now, the complexity of position embedding is upper bounded by $D   (\log D)^2$ \cite{SchurTransform2}.   Since $j$ ranges from 0 to $N/2$,   the total complexity of the position embedding and of its inverse scales  as $    N    ( \log N )^2$. 
 From the above reasoning, it is clear that the bottleneck of the encoding is  the implementation of the Schur transform, which leads to an overall complexity of ${\rm poly}(N)$ for the encoding circuit.  The situation is similar for the decoding, which also uses  position embedding to perform operations depending on  $j$ (see Fig. \ref{fig:decoding}).  The only new  parts  are the initialization  of $N/2  + 1  -  |\set S_\epsilon|$ qubits in the index register and the preparation  of maximally mixed states of rank $m_j$ in the multiplicity register, which  can be approximately generated with exponential precision in $O(N^2)$ operations  \cite{supplemental}.      Summing over the values of $j$ in  $\set S_\epsilon$, we then obtain a number of operations upper bounded by  $O(  N^2)  |\set S_\epsilon|      =  O(  N^{5/2})$. 
  From the above count it is clear that  the overall complexity is polynomial in $N$.    In addition to the computational complexity, it is worth discussing the size of the ancillary systems needed in our compression protocol.  Since the multiplicity register is discarded, the Schur transform in our protocol needs only an ancilla of  $O(\log N)$ qubits \cite{estimation}.   The position embeddings  require   ancillas of size $O(N)$, but, as mentioned earlier,  they can be implemented on a classical computer.  Hence, the total  number of qubits that need to be kept coherent throughout our protocol scales only as $O(\log N)$.



Our compression protocol, presented for qubits, can be generalized to quantum systems of arbitrary dimension  $d$. 
In this case, an ensemble of  $N$ identically prepared rank-$r$  states with known spectrum can be compressed   with error less than $\epsilon$ into approximately $  \left(2dr  -  r^2-1\right)/2 \, \log N$ qubits. In addition, one can take advantage of the presence of degeneracies  and further reduce the number of qubits:    every time the same eigenvalue appears in the spectrum the number of qubits is reduced by at least $1/2\log N$(see \cite{supplemental} for the exact value).   Again, the protocol can be implemented efficiently and is optimal under suitable symmetry assumptions \cite{supplemental}.

In this Letter we showed how to efficiently store  ensembles of  identically prepared quantum systems   into an exponentially smaller memory space.
    For mixed states we discovered that, whenever a  nonzero error is allowed, the size of the  memory  is cut down in a discontinuous way,  provided that the spectrum of the state is known with sufficient precision. 
Intriguingly, the dropoff in the memory size  takes place as soon as the prior information about the eigenvalues is more than the information that could be extracted by a measurement on the input copies.  
 Our approximate  compression protocols  can be implemented efficiently on a quantum computer.

\medskip
{\emph{Acknowledgments.} 
We thank M. Ozols and the referees of this Letter for a number of comments that stimulated  substantial improvements of the original manuscript.  This work is supported by the National Natural Science Foundation of China through Grant No. 11450110096,  by the Foundational Questions Institute (Grant No. FQXi-RFP3-1325),  by the 1000 Youth Fellowship Program of China, and by the HKU Seed Funding for Basic Research.    }

\bibliographystyle{apsrev4-1}
\bibliography{compression} 
\appendix

\begin{widetext}

\section{PROOF OF THEOREM 1}\label{app:opt_zeroerror}
Here we show the optimality of our the  error protocol in the main text. Specifically, we show   that  no zero-error protocol exists that compresses a complete ensemble of mixed states into less than  $ \lceil    2  \log  (N+2)-2\rceil$.

\subsection{The zero error condition}   

The condition for zero-error compression requires that the average error defined as 
\begin{align}\label{error}
e_N= \sum_{\st n} \, p_{\st n}
   \frac{ \left\|\rho^{\otimes N}_{\st n}-\map{D}\circ\map{E}\left(\rho_{\st n}^{\otimes N}\right)\right\|} 2  =  0 \, \, .
\end{align}
This condition immediately implies  $\|\map D\circ\map E(\rho_{\st n}^{\otimes N})-\rho_{\st n}^{\otimes N}\|=0$ for every $\st n$ except for a zero-measure set. Since the Hermitian operator $\map D\circ\map E(\rho_{\st n}^{\otimes N})-\rho_{\st n}^{\otimes N}$ has only zero eigenvalues, it must be a null operator.
 Hence,  the channel  $\map C  :  = \map D \circ\map E$ must fix $\rho_{\st n}^{\otimes N}$, namely that  
\begin{align}\label{above}  
\map C(\rho_{\st n}^{\otimes N})   =  \rho_{\st n}^{\otimes N} 
\end{align}
for every $\st n$ except for a set of zero measure.  Since $p_{\st n}$ has full support on the Bloch sphere, the above condition holds for a dense set of points on the Bloch sphere. As a result,  for every Bloch vector $\st{n}$  there exists a sequence $\left\{\rho_{\st {n}_k}^{\otimes N}\right\}$ of Bloch vectors  satisfying Eq.  (\ref{above})   such that $\lim_{k\to \infty} \st n_k  =\st n  $ and 
$$\lim_{k\to\infty}\rho_{\st{n}'_k}^{\otimes N}=\rho_{\st{n}}^{\otimes N} \, .$$
Consequently, we have
\begin{align*}
\left\|\map D\circ\map E(\rho_{\st n}^{\otimes N})-\rho_{\st n}^{\otimes N}\right\|_1&=\left\|\map D\circ\map E\left(\lim_{k\to\infty}\rho_{{\st n}'_k}^{\otimes N}\right)-\lim_{k\to\infty}\rho_{{\st n}'_k}^{\otimes N}\right\|\\
&=\left\|\lim_{k\to\infty}\left[\map D\circ\map E(\rho_{{\st n}'_k}^{\otimes N})-\rho_{{\st n}'_k}^{\otimes N}\right]\right\|\\
&=0,
\end{align*}
which implies that $\map C(\rho_{\st n}^{\otimes N})   =  \rho_{\st n}^{\otimes N} $ for every vector $\st{n}$ on the Bloch sphere.

\subsection{The algebra associated to the fixed points of a channel}
 Here  we  develop a technique  that generates fixed points of a given channel starting from an initial set of fixed points. Our technique is based on  a result by Blume-Kohout \emph{et  al} \cite{IPS} characterizes  the fixed points. Specifically, Theorem 5 of  Ref. \cite{IPS} guarantees that   one can find a  decomposition  of the  Hilbert space as $ \spc H=  \bigoplus_{k}  \left( \spc{L}_k\otimes \spc{M}_k  \right)$, with the property that the fixed points of a given channel acting on $\spc{H}$ are all the operators of the form  
\begin{align}\label{formfix}
A=\bigoplus_k    \left(  A^{(k)}\otimes  \omega^{(k)}_0 \right)\, ,
\end{align}
where $A^{(k)}$ is an arbitrary matrix on $\spc{L}_k$ and $\omega^{(k)}_0$ is a fixed non-negative matrix on $\spc{M}_k$. 
  Using this fact, we  develop a technique  that generates fixed points of a channel starting from an initial set of fixed points.  
  \begin{prop}\label{prop:fixalgebra}
  Let ${\rm Fix}   (\map C)$ be the set of fixed points of channel $\map C$, let $\{  A_x \}_{x\in\set X}   \subset  {\rm Fix}  (\map C)$ be a subset of non-negative fixed points, and let $\mu (\d x)$ be a non-negative measure on $\set X$.  Then, the set of operators    
   \[    \map A    =  E^{-1/2}  \, {\rm Fix}  (\map C)  \, E^{-1/2}   
 \, , \qquad    E : =   \int  \mu( \d x) \,    \,  A_x  \, ,  \]
 is a matrix $\ast$-algebra  (i.~e.~a matrix algebra closed under adjoint).  Moreover,  one has   $ E^{1/2}    \map A    E^{1/2}  \subseteq  \Fix(\map C)$. 
  \end{prop}
 [Notation:  for a non-invertible operator $E$,  we define $E^{-1}$ as the inverse on the support of $E$.] 
  
\Proof ~  Writing each operator $A_x$ in the form (\ref{formfix}), we obtain 
\[   E   =   \bigoplus_k  \left(    E^{(k)}  \otimes \omega_0^{(k)}\right)  \, , \qquad E^{(k)}   =   \int \mu(\d x)      A_x^{(k)} \, .  \]
Hence, for a generic fixed point  $A\in   \Fix(\map C)$, decomposed as in Eq. (\ref{formfix}), we have 
\[     E^{-1/2}   A  E^{-1/2}  =  \bigoplus_k \left[     \left( E^{(k)} \right)^{-1/2}    A_k        \left( E^{(k)} \right)^{-1/2}  \otimes   P_k  \, ,  \right]\]
where $P_k$ is the projector on the support of $  \omega_0^{(k)}$.  Since  each $A_k$ is a generic operator on $\spc L_k$, we have 
\[   E^{-1/2}   \, \Fix (\map C)  \,  E^{-1/2}  =  \bigoplus_k \left[    {\mathsf B}   ( \spc S_k )  \otimes P_k \right] \, ,   \]
where ${\mathsf B}   ( \spc S_k )$ denotes the algebra of all linear operators on the subspace $\spc S_k  =  \Supp \left[  E^{(k)}\right]$.    Hence,  $\spc A  =   E^{-1/2}   \, \Fix (\map C)  \,  E^{-1/2}$ is an algebra and is closed under adjoint. 
  On the other hand, we have  
\[   E^{1/2}   \, \spc A  \,  E^{1/2}  =  \bigoplus_k \left[    {\mathsf B}   ( \spc S_k )  \otimes M^{(k)}_0 \right] \, ,  \]  
meaning that every operator in  $   E^{1/2}   \, \spc A  \,  E^{1/2}$ is of the form (\ref{formfix})---that is, it is a fixed point.   \qed

\subsection{The minimal algebra required by the zero error condition}

Let us apply  Proposition \ref{prop:fixalgebra} to the channel $\map C  =   \map D\circ \map E$, resulting from the 
concatenation of the encoding and the decoding in a generic zero-error protocol.  By the zero-error condition, all the states $\rho_{\st n}^{\otimes N}$  are fixed points.   
The states can be decomposed as
 \begin{align}\label{statesagain} \rho_{\st n}^{\otimes N} =  \bigoplus_{j=0}^{N/2}  \,  q_{j,N}  \, \left(\rho_{\st n,j}\otimes \frac{I_{m_j }}{m_j}\right) \, .
 \end{align}
 \emph{A priori}, this block decomposition could be completely unrelated with the block decomposition of Eq. (\ref{formfix}).   Proving that the two decompositions coincide will be the main part of our argument.  

Choosing the measure  $\mu(\d x)$ in Proposition \ref{prop:fixalgebra} to be  the invariant measure over $\st n$, the average operator $E$ is given by  
\[ E  =  \bigoplus_{j=0}^{N/2}  \,  q_{j,N}  \, \left(  \frac{I_j}{d_j}\otimes \frac{I_{m_j }}{m_j}\right) \, .    \] 
Hence, the algebra $\spc A$ defined in Proposition \ref{prop:fixalgebra} must contain all the operators of the form  
\[   E^{-1/2}  \, \rho_{\st n}^{\otimes N} \,  E^{-1/2}  =     \bigoplus_{j=0}^{N/2}  \,   \, \left(  d_j  \, \rho_{\st n,j}\otimes I_{m_j } \right)  \, , \]
for every unit vector $\st n$.  Hence, $\map A$  must contain the smallest algebra  $\spc A_{\min}$  generated by the above operators.  
We will now characterize this algebra:  

\begin{prop}\label{prop:top}  
    If the states in Eq. (\ref{statesagain}) are not maximally mixed,  $\spc A_{\min}$  contains the matrix algebra  of all operators on the symmetric subspace, corresponding to $j=  N/2$ in the decomposition (\ref{statesagain}).   
\end{prop}
\Proof  ~  Let us express  the state $\rho   =  p  |0\>\<0|  +  (1-p)  |1\>\<1|$ as $\rho  =   e^{-\beta  Z}/\Tr[   e^{-\beta  Z}]$, $Z  =  |0\>\<0|  -  |1\>\<1|$  for a suitable $\beta\ge 0$.    
By definition,  for every unitary $U  \in  \grp{SU}(2)$,  the algebra $\spc A_{\min}$ contains the operator  
\begin{align}
\nonumber  A_U    &: =      E^{-1/2}     (U\rho  U^{\dag})^{\otimes N}   E^{-1/2}\\
\label{au} &   = \bigoplus_{j=0}^{N/2}  \,       \frac{d_j}{\Tr \left[ e^{-\beta  J^{(j)}_z}\right]}  
\, \left( U^{(j)}  \, e^{-\beta  J^{(j)}_z}U^{(j)\dag} \otimes I_{m_j } \right)  \, ,   
 \qquad \qquad J_z^{(j)}   =  \sum_{m=-j}^j   \,  m  \, |j,m\>\<j,m|    
 \end{align} 
 where $U^{(j)}$ denotes the $(2j+1)$-dimensional irreducible representation of $\grp{SU} (2)$.
Moreover,   since the algebra  $\spc A_{\min}$  is closed under linear combinations, $\spc A_{\min}$ must contain the operator 
\[  X_l  =    \int   \d U  \,  \chi_U^{(l)} ~   A_U  \, , \]
where $\chi_U^{(l)}$ are the characters of the irreducible representations of $\grp{SU}(2)$ given by $  \chi_U^{(l)}  =  \Tr  [  U^{(l)}]$.  
Let us set $l=  N$. In this case, the orthogonality of $\grp{SU}(2)$ matrix elements eliminates all terms in the block decomposition of $\rho^{\otimes N}$, except for the term with $j=N/2$. Notice that in this case the multiplicity subspace is trivial. Hence,  one has  
\begin{align*}
X_N  &  =  \int   \d U  \,  \chi_U^{(N)}  \,         d_{N/2}       \, U^{(N/2)}  \, \rho_{N/2}  \, U^{(N/2)\dag}     \qquad   \qquad  \rho_{N/2}   = \frac{ e^{-\beta   J_z^{(N/2)}}}{ \Tr \left[e^{-\beta   J_z^{(N/2)}}\right]} \, .
\end{align*} 
 The matrix elements of $X_N$ can be computed explicitly as 
\begin{align*}
 \left\< \frac N2,  n \right |       \, X_N  \, \left |\frac N2,   n'\right\>  &  =  \frac{d_{N/2} }{ \Tr \left[e^{-\beta  J^{(N/2)}_z}\right]}\,  \int \d U  \,  \chi_U^{(N)}  \,   \left [   \sum_{m=-N/2}^{N/2}     \,   e^{-\beta  m}    \,  \left\<  \frac N2, n \right |    U^{(N/2)}   \left |\frac N2,m\right\>  \left\< \frac N2 , m\right|    U^{(N/2)\dag}    \left |\frac N2,n'\right\>       \right] \\
&
  =  \delta_{n,n'} \,   (-1)^{n}  \,  \frac{ d_{N/2} \,     \left\<  \frac N2, n,  \frac N2, -n'  | N, 0\right\>}{  d_N \Tr \left[e^{-\beta  J^{(N/2)}_z}\right]}\,     \left [   \sum_{m=-N/2}^{N/2}     \,  (-e^{-\beta })^m     \,        \overline{\left. \left\<   \frac N2,m,  \frac N2, -m\right|  N,0\right\>}  \right] \, \\
&= \delta_{n,n'} \,   (-1)^{n}  \,  \frac{ d_{N/2} \,     \left\<  \frac N2, n,  \frac N2, -n'  | N, 0\right\>}{  d_N \Tr \left[e^{-\beta  J^{(N/2)}_z}\right]}\,     \left [   \sum_{m=-N/2}^{N/2}     \,     \,        \frac{(N!)^2(-e^{-\beta } )^m }{(N/2-m)!(N/2+m)!\sqrt{(2N)!}} \right] \,\\
&=\delta_{n,n'} \,   (-1)^{n+N/2}  \,  \frac{ d_{N/2}(N!)e^{\beta N/2}(1-e^{-\beta})^N \,     \left\<  \frac N2, n,  \frac N2, -n'  | N, 0\right\>}{  d_N \sqrt{(2N)!}\Tr \left[e^{-\beta  J^{(N/2)}_z}\right]}\, ,
 \end{align*} 
 $  \<j_1,m_1, j_2,m_2|  J, M\>$ denoting the Clebsch-Gordan coefficient. Note that the Clebsch-Gordan coefficient in the above expression is nonzero if and only if $n=n'$.  As a consequence, the operator $X_N$ has full support.   

Now, since $\spc A_{\min}$ is  an algebra, it must  contain $X_N$  as well as the whole Abelian algebra generated by it.   In particular, it must contain the projector on the support of $X_N$---which is nothing but $ P_{N/2}$, the projector on the symmetric subspace.   Moreover, it must contain all the operators of the form   \[ A_{U,N/2}    = P_{N/2}  A_U  P_{N/2}    \propto   U^{(N/2)}  \, e^{-\beta   J^{(N/2)}_z}  \, U^{(N/2)\dag}   \qquad \forall\,U  \in  \grp{SU}(2) \,.  \]
Finally,    for $\beta  \not  =  0$, it is easy to see that  the smallest algebra   $\map A_{\min, N/2}$ containing the above operators is the algebra $\mathsf{B}   (\spc R_{N/2})$.   This can be easily seen by von Neumann's double commutant theorem:    If an  operator  $B$ commutes with the non-degenerate Hermitian operator $A_{U,N/2}$ for every $U$, then $B$ must be proportional to the identity.   Hence, the double commutant  of $\map A_{N/2}$---equal to $\map A_{N/2}$ itself---is   the whole  $\mathsf{B}   (\spc R_{N/2})$.    In conclusion, we have the inclusion      $\mathsf{B}   (\spc R_{N/2})  \subseteq   \spc  A_{\min, N/2}  \subseteq \spc A_{\min}$.   \qed

\begin{prop}\label{prop:todos}  If the states in Eq. (\ref{statesagain}) are neither pure nor  maximally mixed, then $\spc A_{\min}$ is the full algebra generated by the $N$-fold tensor representation of $\grp{GL}(2)$, namely   
 \[  \spc A_{\min}  =   \bigoplus_{j=0}^{N/2}   \, \left[  \mathsf{B}  (  \spc R_j)  \otimes I_{m_j }\right] \, ,\]
$\mathsf{B}  (  \spc R_j)  $ denoting the algebra of all linear operators on the representation space $ \spc R_j$.    
\end{prop}
\Proof ~ We prove that $\spc A_{\min}$ contains the algebra  $  \mathsf{B}  (  \spc R_j)  \otimes I_{m_j }$ for every $j$.    The proof is by induction, with $j$  starting from $N/2$ and going down  to $0$.  For $j=N/2$ we know that $\spc A_{\min}$ contains the algebra  $\mathsf{B}   (\spc R_{N/2})$ of all operators with support in the symmetric subspace.  Let us assume that $\spc A_{\min}$ contains all the algebras      $\mathsf{B}  (  \spc R_j)  \otimes I_{m_j }$ with $j  \ge  j_*+1$ and show that it must necessarily contain also the algebra           $\mathsf{B}  (  \spc R_{j_*})  \otimes I_{m_{j_*} }$.   
By construction, we know that  $\spc A_{\min}$ contains all the operators $A_U$ of the form 
\begin{align}
\nonumber  A_U      = \bigoplus_{j=0}^{N/2}  \,       \frac{d_j}{\Tr \left[ e^{-\beta  J^{(j)}_z}\right]}  
\, \left( U^{(j)}  \, e^{-\beta  J^{(j)}_z}U^{(j)\dag} \otimes I_{m_j } \right)  \, ,   
 \qquad \qquad J_z^{(j)}   =  \sum_{m=-j}^j   \,  m  \, |j,m\>\<j,m|   \, . 
 \end{align} 
Since the states  in Eq. (\ref{statesagain}) are not pure, all the blocks in the sum are non-zero.  
  Moreover,  the induction hypothesis implies that $\spc A_{\min}$ should also contain the operators $A_U'$ of the form  
  \[  A_U  '       =   \bigoplus_{j=0}^{j_*}  \,    \frac{  d_j}{\Tr \left[ e^{-\beta J^{(j)}_z}\right]}     \, \left(U^{(j) }   \, e^{-\beta  J^{(j)}_z}   U^{(j)  \dag}   \otimes I_{m_j } \right)   \, , \qquad   U  \in  SU(2) \, .  \]
 Now, we can repeat the argument used in the proof of Proposition \ref{prop:top}:  by linearity, $\spc A_{\min}$ must contain the operator 
\begin{align*}  X_{2j_*}    &=    \int   \d U  \,  \chi_U^{(2j_*)} ~   A'_U \\
&   =        \frac{  d_{j_*}}{\Tr \left[ e^{-\beta J^{(j_*)}_z}\right]}         \int   \d U  \,  \chi_U^{(2j_*)} ~  \, \left(    U^{(j_*)}  \, e^{-\beta  J^{(j_*)}_z} \,  U^{(j_*)  \dag} \otimes I_{m_{j_*} } \right)    \, . 
\end{align*}
Explicit calculation   (same as in Proposition \ref{prop:top}) shows that $X_{2j_*}$ has full rank.  
Hence, the projector on the support of $X_{2j_*}$ is $ P_{j_*}  =   I_{j_*}  \otimes I_{m_{j_*}}$.  Since $\spc A_{\min}$ should contain this projector, it must also contain all operators of the form  
\begin{align*}  A'_{U,j_*}  & =    P_{j_*}  A'_U  P_{j_*}  \\
&  \propto     U^{(j_*)}   e^{-\beta   J^{(j_*)}_z } U^{(j_*) \dag}  \otimes I_{m_{j_*}}  \, , \qquad U\in SU(2) \, .
\end{align*} 
Again, using von Neumann's double commutant theorem, it is easy to show that the smallest algebra containing all the above operators is   $  \mathsf{ B} (\spc R_{j^*})  \otimes I_{m_{j_*}}$.    
In conclusion we proved that $\spc A_{\min}$ must contain $  \mathsf{ B} (\spc R_{j^*})  \otimes I_{m_{j_*}}$.  By induction, this proves the inclusion
\[  \spc A_{\min}  \supseteq  \bigoplus_{j=0}^{N/2}   \, \left[  \mathsf{B}  (  \spc R_j)  \otimes I_{m_j }\right] \, .\]
In the other hand, the definition of $\spc A_{\min}$ implies the opposite inclusion. Hence, one must have the equality.  
 \qed 

\subsection{Zero-error compression of a complete ensemble implies zero error compression for every ensemble of permutationally invariant states}  

Propositions  \ref{prop:fixalgebra} and \ref{prop:todos} imply  the following 
\begin{cor}\label{cor:done}
If the states (\ref{statesagain}) are neither pure nor maximally mixed, every channel $\map C$ preserving them must preserve all permutationally invariant states. 
\end{cor}

\Proof    ~ By Propositions  \ref{prop:fixalgebra} and \ref{prop:todos}, the channel $\map C$ must  satisfy   
\[  \Fix (\map C)  \supseteq    \spc A_{\min}   =  \bigoplus_{j=0}^{N/2}  \,     \left[   \mathsf B    (    \spc R_j)  \otimes I_{m_j}   \right]  \, ,\]
meaning that the full  algebra generated by the tensor representation of $\grp{GL}(2)$ is contained in the set of fixed points.  \qed

\medskip  

We are now in position to prove Theorem 1 in the main text: 

\medskip 

{\bf Proof of Theorem 1.}  Suppose that a compression protocol has zero error on a complete ensemble of mixed states.  
Then, Corollary \ref{cor:done}  implies that the protocol should have zero error on all permutationally invariant states.   
In particular, the protocol should be able to transmit without error the following ensemble of orthogonal pure states 
\[S:=  \left.  \left\{ \rho_{j,m}  =   |j,m\>\<  j,m|  \otimes \frac{I_{m_j}}{m_j}  , p_{j,m}=\frac{1}{D}\, \right  |\,      j  =  0, \dots, N/2  \, ,   m  =  -j, \dots, j, \,D:=\sum_j d_j \right\} \, .\] 
 A lower bound on the dimension $d_{\rm enc}$ of the encoding space $\spc H_{\rm enc}$ is then obtained by considering the amount of classical information carried by $S$. In detail, the lower bound can be calculated using the monotonicity of Holevo's chi quantity in quantum data processing. Holevo's chi quantity of $S$ \cite{holevo} is defined as follows
\begin{align*}
\chi\left(S\right)&:=H\left(\sum_{j,m} p_{j,m}\rho_{j,m}\right)-\sum_{j,m}p_{j,m}H\left(\rho_{j,m}\right)
\end{align*}
with $H(\rho)$ being the von Neumann entropy of the state $\rho$. Since the chi quantity is non-increasing under quantum evolutions, in the zero-error scenario we have
\begin{align}\label{chieq1}
\chi\left(S\right)=\chi\left(S_{\rm enc}\right)
\end{align}
where $S_{\rm enc}$ is the encoded ensemble $S_{\rm enc}:=\{\map E(\rho_{j,m}), p_{j,m}\}$. On the other hand, the dimension of the encoding subspace is lower bounded by the chi quantity \cite{Horodecki}
\begin{align}\label{chibound1}
\log d_{\rm enc}\ge \chi\left(S_{\rm enc}\right).
\end{align}

The chi quantity for the ensemble $S$ can be computed as
$\chi\left(S\right)=\log D \, .$
Combining this equality with Eqs. (\ref{chieq1}) and (\ref{chibound1}) we get $$d_{\rm enc}\ge D=\left(\frac N2+1\right)^2,$$ which concludes the optimality proof.   The protocol showed in the main text saturates the bound.  
\qed
 
\section{PROOF OF THEOREM 2}

As stated in the main text, we assume $p>\frac12$, because for $p=1/2$ the ensemble is trivial, consisting only of the maximally mixed state.  

We first notice that the error of the compression protocol  is upper bounded as 
\begin{align}
e_{N}&  =  \frac 12    \,  \left \|    \rho_{\st n}^{\otimes N}   -       \map D  \circ \map E  \left ( \rho_{\st n}^{\otimes N} \right) \right\|  \, ,  \qquad    \forall   \st n    \in   \mathbb S^2    \nonumber  \\
     &  =  \frac 12 \,   \left \|   \sum_{j  \not \in  \set S_\epsilon  }   \,    q_{j,N}     \,   \left[    \rho_{\st n ,   j} \otimes  \frac{ I_{m_j}}{m_j}     -   \map D(  \rho_0 )  \right]  \right \|\nonumber\\
  &\le            \sum_{j\not\in\set S_\epsilon}q_{j,N}    \, \label{errorbound},
\end{align}
the last step following from the triangle inequality and from the fact that the trace distance of two states is upper bounded by 2.    Note that the upper bound is independent of $\st n$, meaning that the protocol works equally well for all states with the same spectrum (or equivalently, for all states with the same purity).

At this point, it is enough to prove that  the upper bound  vanishes in the large $N$ limit. To this purpose, we  use the expression for $q_{j,N}$  [Eq. (5) in the main text] and observe that one has
\begin{align}\label{F2}
1-e_N\ge&\sum_{j\in\set{S}_\epsilon}\frac{2(2j+1)}{j_0}B\left(N+1,p,\frac{N}{2}+j+1\right)-\sum_{j\in\set{S}_\epsilon}\frac{2(2j+1)}{j_0}B\left(N+1,p,\frac{N}{2}-j\right)
\end{align}
where $j_0=(2p-1)(N+1)/2$. The second summand in  the r.h.s. of Eq. (\ref{F2}) is negligible in the large $N$ limit: precisely, it can be bounded as
\begin{eqnarray}
\sum_{j\in\set{S}_\epsilon}\frac{2(2j+1)}{j_0}B\left(N+1,p,\frac{N}{2}-j\right)&\le&\sum_{j=0}^{\frac{N}{2}}\frac{2(2j+1)}{j_0}B\left(N+1,p,\frac{N}{2}-j\right)\nonumber\\
&\le&\frac{1}{2p-1}\sum_{j=0}^{\frac{N}{2}}B\left(N+1,p,\frac{N}{2}-j\right)\nonumber\\
&\le&\frac{1}{2p-1}\exp\left[-\frac{2(2p-1)^2N^2}{N+1}\right]\label{F4}
\end{eqnarray}
having used the Hoeffding's inequality in the last step.  Hence, this term  goes to zero exponentially fast with $N$,

Now, recall that  we chose  $\set S_\epsilon$ to be the interval
\begin{align}\label{Sepsilon}
\set{S}_\epsilon=\left[j_0-1/2-\sqrt{N\ln(2/\epsilon)},j_0-1/2+\sqrt{N\ln(2/\epsilon)}\right].
\end{align}

Setting
$j_0-j-1/2= x$, we then obtain  
\begin{align*}
e_N&\le1-\sum_{x=-\sqrt{N\ln(2/\epsilon)}}^{\sqrt{N\ln(2/\epsilon)}}\left(1-\frac{x}{j_0}\right)B\left(N+1,p,p(N+1)-x\right)+\frac{1}{2p-1}\exp\left[-\frac{2(2p-1)^2N^2}{N+1}\right]\\
&=1-\sum_{x=-\sqrt{N\ln(2/\epsilon)}}^{\sqrt{N\ln(2/\epsilon)}}B\left(N+1,p,p(N+1)-x\right)+\frac{1}{2p-1}\exp\left[-\frac{2(2p-1)^2N^2}{N+1}\right]\\
&\le 2\exp\left[\frac{2N}{N+1}\ln\frac{\epsilon}{2}\right]+\frac{1}{2p-1}\exp\left[-\frac{2(2p-1)^2N^2}{N+1}\right]\\
&\le \epsilon^{\frac{2N}{N+1}}+\frac{1}{2p-1}\exp\left[-\frac{2(2p-1)^2N^2}{N+1}\right]
\end{align*}
In the second last step we have used the Hoeffding's inequality. Now it can be seen that the right hand side of the bound vanishes exponentially fast with $N$, and we can always find a $N_0$ such that $e_N\le\epsilon^{3/2}<\epsilon$ for any $N>N_0$.   The dimension of the encoded system is now 
\begin{eqnarray*}
d_{\rm enc}&=&\sum_{j\in\set{S}_\epsilon}(2j+1)\\
&=&2(2p-1)\sqrt{N\ln(2/\epsilon)}(N+1)
\end{eqnarray*}
An upper bound on the number of required qubits is given by 
\begin{align*}
\log d_{\rm enc}&=\log\left[2(2p-1)N\sqrt{N\ln\frac{2}{\epsilon}}\right]+\log\left(1+\frac1N\right)\\
&\le\frac32 \log N+\log\left[2(2p-1)\sqrt{\ln\frac2\epsilon}\right]+1\\ 
\end{align*}
\qed

\section{THE PURE STATE CASE: NO DISCONTINUOUS GAP BETWEEN ZERO-ERROR AND APPROXIMATE COMPRESSION}
Here we prove that the type of discontinuity highlighted by our Theorems 1 and 2 is specific to mixed states.  
  Consider the pure state ensemble
     $\left\{ \left(  |\st n\>\<\st n| \right)^{\otimes N}\,   ,  \d^2 \st n\right\}$, where $  |\st n\>$ is the pure qubit state with Bloch vector $\st n$ and $\d^2\st n$ is the invariant measure on the Bloch sphere.  
  Suppose that the state $\left(  |\st n\>\<\st n|\right)^{\otimes N}$ is encoded into a state   $\rho_{\st n,{\rm enc}}$ on a Hilbert space of dimension $d_{\rm enc}$.   Assuming that the compression error is bounded by $\epsilon$, an argument by Horodecki  \cite{Horodecki}  gives a lower bound on $d_{\rm enc}$.  
  The argument is based on the following lemma, based on  the Alicki-Fannes inequality 
   \begin{lem}[\cite{alicki}]
Let $\{  \rho_x\, ,  p_x\}$ be an ensemble of states and let $\{  \rho_{x,{\rm enc}} \, ,  p_x\}$ be the ensemble of the encoded states.  If the  compression protocol has error bounded by $\epsilon$, then the following inequality holds
\begin{align}\label{chi bound}
\left|\chi\left(\left\{ \rho_x \, ,p_x  \right\}\right)  -\chi\left(\left\{ \rho_{x,{\rm enc}} \, , p_x\right\} \right) \right|   \le 2\left[\epsilon\log d_{\rm in}+\eta(\epsilon)\right],
\end{align}
where  $d_{\rm in}$ is the rank of the average state $\rho  = \sum_x \, p_x  \rho_x $ and 
$\eta(x)=-x\ln x$. 
\end{lem}
In our case, $d_{\rm in}$ is the dimension of the symmetric subspace, namely 
\begin{align}\label{din}
d_{\rm in} = d_{\frac{N}{2}}=N+1 \,.
\end{align}
Moreover, we have     
\begin{eqnarray}\label{chi}
\chi\left(\left\{ \left(  |\st n\>\<\st n|\right)^{\otimes N}  , \,   \d^2 \st n  \right\}  \right)=H\left(I_{\frac{N}{2}}/d_{\frac{N}{2}}\right)=\log (N+1).
\end{eqnarray} 
and, by  the Holevo's bound \cite{holevo}, 
\begin{align}\label{hb}
\chi\left(\left\{  \rho_{\st n, {\rm enc}}  \, ,  \d^2 \st n  \right\}  \right)  \le \log d_{\rm enc} \, .
\end{align}  
In our case, we have $d_{\rm in}=d_{\frac{N}{2}}=N+1$. Hence, combining Eqs.  (\ref{chi bound}),  (\ref{din}), (\ref{chi}), and (\ref{hb}) we obtain the bound
\begin{align*}
\log d_{\rm enc}&\ge (1-2\epsilon)\log(N+1)-2\eta(\epsilon).
\end{align*}
Now, note that   the r.h.s. is continuous in $\epsilon$ and tends to $\log (N+1)$ when $\epsilon$ tends to zero.    The value  $\log (N+1)$ is exactly the minimum number of qubits needed to encode a generic state in the symmetric subspace with zero error.   Hence, as $\epsilon$ tends to zero, the number of qubits needed for approximate compression tends to the number of qubits needed for zero-error compression.

\section{PROOF OF THEOREM 3}\label{app:qubit-coverse}

Here we prove the optimality of our protocol among all compression protocols where the  encoding is covariant and the decoding preserves the magnitude of the  total angular momentum.    Precisely,  we assume that   
 \begin{enumerate}
 \item the encoding space $\spc H_{\rm enc}$  supports a unitary representation of the group $\grp {SU}  (2)$, denoted by  $  \{  V_g ~|~ g\in\grp {SU} (2) \}$
 \item  the encoding channel satisfies the covariance condition     
 \begin{align}\label{covariance}  
 \map E \circ    \map U_g      =   \map  V_g  \circ  \map E   \, , \qquad \forall  g\in\grp {SU}  (2)   \, ,
 \end{align} 
 where $\map U_g$ and $\map V_g$ are the unitary channels defined by $\map U_g  (\cdot):  =  U_g  \cdot  U_g^\dag $ and $\map V_g=  V_g \cdot  V_g^\dag$. 
 \item  the decoding channel $\map D$   preserve the magnitude of the total angular momentum, in the sense that, for every input state $\rho$, one has 
  \begin{align}\label{conservation}
   \Tr  \left [ \st K^2   \,    \map D( \rho)  \right]  =   \Tr  \left[   \st J^2  \,  \rho  \right ]        \, ,
 \end{align}
where    $\st  K=  (  K_x,K_y,K_z)$ are the generators of the representation $\{  V_g  \, ,   g \in  \grp {SU}  (2)\}$  and $ \st J  =  (  J_x,  J_y, J_z)$ are the generators of the representation $\{U_g^{\otimes N} ,     g\in\grp {SU}  (2)\}$.   
 \end{enumerate}
 
Under these conditions, we can prove the optimality of the protocol presented in Theorem 3 of the main text.  

\medskip  

{\bf Proof of Theorem 3.}   For the purpose of this proof, it is convenient to parametrize the mixed states $\rho_{\st n}$ as $\rho_g   =    U_g  \rho  U_g^\dag$, where  $\rho$ is a fixed state and $g$ is a generic element of $\grp { SU}  (2)$.     Let us decompose the encoding space as   
\begin{align}\label{iso}  
\spc H_{\rm enc}     =  \bigoplus_j  \,   \left(  \spc R_j  \otimes \widetilde {\spc M}_j \right) \, ,
\end{align}   
where $j$ is the quantum number of the angular momentum, $ \map R_j$ is the corresponding representation space, and $\widetilde{\spc M}_j$ is a suitable multiplicity space.   By definition, one has  
\begin{align}
\nonumber   \spc H_{\rm enc}     & \supseteq  \Span  \left\{   \Supp  \left [ \map E  \left(\rho_g^{\otimes N}\right) \right]    ,     g\in\grp {SU}  (2)   \right\}  \\
\label{inclusion}  &    =  \Span  \left[   \Supp  \left (  \Omega \right) \right]   \, ,  \qquad \Omega : =   \int \d g \,   \map E \left(\rho_g^{\otimes N}\right)   \, .
\end{align}
Since $\map E$ is covariant, the state $\Omega$ satisfies the relation  $V_g \Omega V_g^\dag   =  \Omega\, , \forall g\in\grp {SU}  (2)$.  Hence, $\Omega$ can  be written in the block diagonal form 
\[\Omega  =  \bigoplus_{j\in  \set S}  \left ( \frac {  I_j}{d_j}   \otimes \omega_j \right)\, , \]
where $\omega_j$ is a suitable state on the multiplicity space and   $\set S$ is a suitable set of values of the angular momentum number.    Combining the above decomposition  with Eq. (\ref{inclusion}), we obtain  the bound
\begin{align}\label{dencbound}   d_{\rm enc}   \ge   \rank  \, \Omega   \ge    \sum_{  j\in \set S}    \,  d_j  \,.     
\end{align}

On the other hand, since the decoding preserves the magnitude of the angular momentum, one has 
 \begin{align*}
\Tr  [ \Pi_j     \,   \map D  \circ\map E   \left(\rho_g^{\otimes N}\right)]   &  =   \Tr  [  \widetilde \Pi_j     \map E   \left(\rho_g^{\otimes N}\right) ]  \,  ,  \qquad     \forall  j   =  0, \dots,  N/2  \, ,    \forall g \in\grp{SU}  (2) \, ,
\end{align*}
where $\Pi_j$ is the projector on $\spc R_j\otimes \spc M_j$ while $\widetilde \Pi_j$ is the projector on $\spc R_j\otimes \widetilde {\spc M}_j$.     Hence, we have  
\begin{align}
\sum_{j\in\set S}   \Tr[  \Pi_j \map D \circ \map E\left(  \rho_g^{\otimes N}\right)]   = 1  \, ,  \qquad   \forall g \in\grp{SU}  (2) \, , 
\end{align}
meaning  that all the output states $\map D \circ \map E\left(  \rho_g^{\otimes N}\right)$ are contained in the subspace $\spc H_N  :  =  \bigoplus_{j\in\set S}   \left(\spc R_j  \otimes \spc M_j\right)$.     Hence, we have 
\begin{align}
\nonumber e_N&=  \frac 12  \left \|  \rho_{g}^{\otimes N}   -    \  \map D \circ  \map E  \left(\rho_{g}^{\otimes  N}  \right)      \right\|    \qquad \forall g \in \grp {SU}  (2)\\
\nonumber &  \ge \frac 12      \left \|    P_N \left[  \rho_{g}^{\otimes N}   -    \  \map D \circ  \map E  \left(\rho_{g}^{\otimes  N}  \right)    P_N  \right]   \right \|   +  \frac 12    \left \|   (I^{\otimes N}  -  P_N)  \left[   \rho_{g}^{\otimes N}   -    \  \map D \circ  \map E  \left(\rho_{g}^{\otimes  N}  \right)  \right]    (I^{\otimes N}  -   P_N)   \right\|      \\
   \nonumber       &=   \frac 12    \left \|   (I^{\otimes N}  -  P_N)   \rho_{g}^{\otimes N}      (I^{\otimes N}  -   P_N)   \right\|       \\  
&\ge \sum_{j\not\in\set{S}}\frac{q_{j,N}}{2} \label{approxt}
\end{align}
where $P_N$ is the projector on $\spc{H}_N$.
Now we prove that any protocol with $d_{\rm enc}=O\left(N^{3/2-\delta}\right)$,  $\delta  >0$,    will have a non-vanishing error.  Recall from the main text that the probability distribution $q_{j,N}$ can be expressed as
\begin{align}\label{dist}
q_{j,N}= \frac{2j+1}{2j_0}   &\left[   B\left(N+1,p,\frac N  2  + j+1   \right)  -B\left(N+1,p,\frac  N2  -  j \right)\right]
\end{align}
where $B(n,p,k)$ is the binomial distribution  with $n$ trials and with probability $p$ and $$j_0 = (p-1/2)(N+1) \, .$$
Combing Eq. (\ref{approxt}) with Eq. (\ref{dist}), we have
\begin{align*}
e_N&\ge\frac12-\frac12\sum_{j\in\set{S}}\frac{2j+1}{2j_0}B\left(N+1,p,\frac{N}{2}+j+1\right).
\end{align*}

We split the set $\set S$ into two subsets $\set S_1$ and $\set S_2$, defined as 
\begin{align*}
\set S_1&=  \set S  \cap \left [ j_0-\frac{\sqrt{cN}+1}2  , j_0+\frac{\sqrt{cN}+1} 2\right  ] \\
\set S_2&=\set S\setminus \set S_1
\end{align*}
where $c$ is an arbitrary  constant. 
 The error is then bounded as
\begin{align}\label{e-bound1}
e_N&\ge\frac12  \left(  1  - s_1  -s_2\right)  \,  \qquad     s_k   : =  \sum_{j\in\set{S}_k}\frac{2j+1}{2j_0}B\left(N+1,p,\frac{N}{2}+j+1\right) \, ,  ~  k  =  1, 2  \, .
\end{align}
We now bound $s_1$ and $s_2$.    Let us start from $s_1$: by definition, we have 
\begin{align}
\nonumber s_1  &\le  \frac{\max_{j\in\set S_1}   ( 2j+1) }{2j_0}\,   \sum_{j\in\set{S}_1}  B\left(N+1,p,\frac{N}{2}+j+1\right)    \nonumber\\
\nonumber &=  O(1)  \,  \sum_{j\in\set{S}_1}  B\left(N+1,p,\frac{N}{2}+j+1\right)    \\
&  \le O(1) \,    |\set S_1|  B\left(N+1,p,\frac{N}{2}+j_0+1\right) \nonumber\\
\label{baab} & =     O\left(  N^{-1/2}\right)\,  |\set S_1|\, .
\end{align}
In turn, $\set S_1$ can be bounded from the relation  
\begin{align}
\nonumber |\set S_1| \,  \left(  \min_{j\in \set S_1}  \, 2j+1 \right)   &  \le    \sum_{j\in\set S_1}    (2j+1)  \\
\nonumber &   \le d_{\rm enc}   \\
&  =   O\left  (    N^{3/2  -\delta}\right) \, ,
\end{align}
which implies $  |\set S_1  |  \le   O( N^{1/2-\delta})$.  
Inserting this relation into Eq. (\ref{baab}),   we finally obtain 
\begin{align}
\label{e-bound2}  s_1  \le        O\left(    N^{-\delta}\right) \, . 
\end{align}

Regarding $s_2$, we have the bound
\begin{align}
s_2  
&\le \frac{N+1}{j_0}\left[\sum_{j\le      j_0   -  \frac {   \sqrt {cN}  +1} 2   }  \,  B\left(N+1,p,\frac{N}{2}+j+1\right)\right]  \nonumber\\
&=\frac{1}{p-1/2}   \left[ \sum_{j\le      j_0   -  \frac {   \sqrt  {c N }  +  1}2 }     \, B\left(N+1,p,\frac{N}{2}+j+1\right)  \right] \nonumber\\
&\le \frac{e^{-c/2}}{p-1/2} \, ,\label{e-bound3}
\end{align}
the last inequality coming  from  Hoeffding's bound. 	

Finally, combining the inequalities (\ref{e-bound1}), (\ref{e-bound2}), and (\ref{e-bound3}), we obtain the lower bound 
\begin{align*}
e_N \ge \frac12   \left[  1  -  O\left(N^{-\delta}\right)-\frac{e^{-c/2}}{p-1/2}  \right] \, ,
\end{align*}
Since the constant $c$ is arbitrary, the bound  becomes  $e_N  \ge 1/2  -  O\left(N^{-\delta}\right)$. \qed




\section{UPPER BOUND ON THE COMPLEXITY OF GENERATING APPROXIMATE MAXIMALLY MIXED STATES}

The decoding requires the preparation of maximally mixed states to be placed in the multiplicity register. For a given value of $j$, this is accomplished by generating a maximally entangled state of rank $m_j$.  
In the following we present a three-step protocol for this purpose. 
\begin{enumerate}
\item Choose an integer $n=O(N)$ such that  $m_j \in(2^{n-1}, 2^n]$. Prepare $n$ maximally entangled qubit states. The resulting the state is $\rho=[|\Phi^+ \rangle \langle \Phi^+ |]^{\otimes n}$, with $|\Phi^+ \rangle = (|00\rangle + |11\rangle )/\sqrt{2}$  and lies in a space of dimension $2^{2n}$. 
\item Perform the measurement in the computational basis on one qubit of each entangled pair. The measurement outcomes of the individual qubit measurements are saved in a sequence of $n$ binary digits, let us denote it by $\underline y$.
\item Compare the string $\underline y$  with the binary expression of $m_j$. If $\underline y$, as a number, is larger than $ m_j$, the protocol fails and we have to restart by preparing again $n$ maximally entangled qubits. Otherwise, we keep the remaining qubits, which, on average, will be in a maximally entangled mixed state of rank $m_j$.
\end{enumerate}
The last step can be seen by noting down the  quantum operation  $\mathcal{C}_{\rm yes}$  corresponding to the successful outcomes of the projective measurement, given by 
\begin{align} 
\mathcal{C}_{\rm yes}(\sigma) = \sum_{y\leq m_j} |\underline{y}\rangle \langle \underline{y} | \sigma  |\underline{y}\rangle \langle \underline{y} | \nonumber \ . 
\end{align}
 The protocol is successful in more than half of the cases. For that reason, the probability of  failure vanishes exponentially in the number of repetitions $l$ as $p_{\rm no} \leq 2^{-l}$. To ensure that the error is vanishing fast enough with the number of state copies $N$, we repeat the protocol $N$ times. Then, the complexity of the protocol is comprised of preparing the qubit states, which takes $O(N)$ steps, and from comparing the $n$ digit binary strings on a classical computer, which also takes $O(N)$ steps. By repeating the protocol $N$ times, the overall complexity yields $O(N^2)$. It is safe to run the protocol $N$ times to assure for an exponentially vanishing error, because the complexity of the decoding is still dominated by the Schur transform.

\section{ZERO-ERROR COMPRESSION FOR  QUANTUM SYSTEMS OF DIMENSION  $d>2$}
In this and the following sections, we generalize our results to quantum systems of arbitrary finite dimension  $d<\infty$.  

\subsection{Upper bound on the number of encoding qubits}  

\begin{theo}\label{thm2}
In dimension $d$, every  ensemble of $N$ identically prepared mixed states of rank $r$     can be  encoded without error into less than  $\left(2dr-r^2+r-2\right)/2 \, \log (N+d-1)$  qubits.  
\end{theo}
  
 The proof is based  on the  Schur-Weyl duality, which   
allows one to decompose the $N$-copy Hilbert space
as 
\[  \spc H^{\otimes N}  \simeq  \bigoplus_{ \lambda   \in   \spc Y_{N,d}}   \, \left( \spc R_{ \lambda}\otimes \spc M_{ \lambda}  \right) \, , \] 
where $\spc R_{ \lambda}$ is a representation space, $\spc M_{ \lambda}$ is a multiplicity space, and the sum runs over the set   $  \spc Y_{N,d}$ of all Young diagrams of $N$ boxes arranged in $d$ rows, parametrized as $  \lambda =  (\lambda_1,\dots, \lambda_d)$,  with $\lambda_1  \ge \lambda_2\ge\dots\ge  \lambda_d$,  $\sum_{i=1}^d \lambda_i  =  N$. We use the notations
$$d_\lambda=\dim \spc R_\lambda$$
and
$$m_\lambda=\dim \spc M_\lambda.$$

Relative to this decomposition,    every state of the form $\rho^{\otimes N}$ where $\rho$ has rank $r$  can be cast into the form 
$$\rho^{\otimes N}  =\bigoplus_{\lambda\in\mathcal{{Y}}_{N,r}}q_{\lambda,N} \left(\rho_{\lambda}\otimes\frac{I_{m_{\lambda}}}{m_\lambda}\right)  \, ,$$
where   $\rho_{\lambda}$ is a quantum state  on  $\spc R_\lambda$,  $I_{m_\lambda}$ is the identity on  $\spc M_\lambda$,   and $q_{\lambda,N}$  is a suitable probability distribution.  Note that only the Young diagrams with    $r$  rows or less are present here  (for this fact, see e.g. \cite{ExactDistribution}).

\medskip 

The proof of Theorem   \ref{thm2}  makes use of the following lemmas:  

\begin{lem}\label{bound-R}
For every $\lambda\in\mathcal{{Y}}_{N,r}$, one has
 $d_\lambda\le (N+d-1)^{(2dr-r^2-r)/2}$.    
\end{lem}
\Proof 
The dimension  can be expressed as  
\begin{align}\label{dimensione} 
d_\lambda  =   \frac{\prod_{1\le i<j\le d}(\lambda_i-\lambda_j-i+j)}{\prod_{k=1}^{d-1}k!}  \,,  
\end{align}
cf. Eq. (III.10) of \cite{unitary-groups}.  Since   $\lambda_i=0$ for $i>r$, we have the following chain of (in)equalities
\begin{align*}
d_\lambda&= \frac{\prod_{1\le i<j\le r}(\lambda_i-\lambda_j-i+j)   \,  \cdot \, \prod_{1\le i\le r<j\le d}(\lambda_i  -i+j)  \, \cdot \, \prod_{r<i<j\le d}(j-i)}{\prod_{k=1}^{d-1}k!} \\
&\le \frac{(N+r-1)^{r\choose 2}  \,  \cdot \,  (N+d-1)^{r(d-r)}  \, \cdot  \,   \prod_{  l=1}^{d-r-1}    l!  }{\prod_{k=1}^{d}k!}
\\
&\le \frac{(N+d-1)^{(2dr-r^2-r)/2}}{ \prod_{  k=d-r}^{d-1}    k!  } \, .
\end{align*}
\qed

\begin{lem}\label{denc}
The total dimension of all the representation spaces corresponding to Young diagrams with no more than $r$ rows is upper bounded as  
\[  \sum_{\lambda\in\mathcal{{Y}}_{N,r}}     \,  d_\lambda <     (N+d-1)^{\frac{2dr-r^2+ r-2}2} \, .  \]    
\end{lem}
  
\Proof  By Lemma \ref{bound-R} one has  
\begin{align*}
  \sum_{\lambda\in\mathcal{{Y}}_{N,r}}     \,  d_\lambda  & \le      (N+d-1)^{\frac{2dr-r^2-r}2}   \,    \left|   \mathcal{{Y}}_{N,r}   \right|   \\  
  &  <          (N+d-1)^{\frac{2dr-r^2+r-2}2}    \, ,
\end{align*} 
having used the equality  $    \left|   \mathcal{{Y}}_{N,r}   \right|   =     {N+r-1\choose r-1} $  \cite{GW98} and the elementary bound  $   {N+r-1\choose r-1}  <  (N+1)^{r-1}  \le  (N+  d-1)^{r-1} $.  
\qed  

\medskip  

{\bf Proof of Theorem \ref{thm2}.}   A zero-error compression protocol is given by the following encoding and decoding channels: 
\begin{align*}
\map E     (\rho)     & =    \bigoplus_{\lambda  \in   {\spc Y}_{N,r}}      \Tr_{\spc M_\lambda}   [\Pi_\lambda  \rho  \Pi_\lambda] \\
\map D(\rho')      &    =      \bigoplus_{\lambda  \in  {\spc Y}_{N,r}}       P_\lambda  \rho'   P_\lambda  \otimes \frac {  I_{m_\lambda}}{m_{\lambda}}  \, ,        
\end{align*}
where $\Pi_\lambda$ is the projector on  $\spc R_\lambda \otimes \spc M_\lambda$ and  $  P_\lambda$ is the projector on $ \spc R_\lambda$.      The encoding space is  $\spc H_{\rm enc}   =  \bigoplus_{\lambda \in  {\spc Y}_{N,r}}     \spc R_\lambda$
and has dimension     $d_{\rm enc}=\sum_{\lambda\in {\spc{Y}}_{N,r}}d_{\lambda}$, which we bound as  
\begin{align*} 
d_{\rm enc}  &=\sum_{\lambda\in {\spc{Y}}_{N,r}}d_{\lambda}    \\  
   &<   (N+d-1)^{\frac{2dr-r^2+r-2 } 2}  \, , 
 \end{align*}
having used Lemma \ref{denc}.  \qed      

\subsection{Lower bound on the number of encoding qubits used by the zero-error protocol}

Here we give a lower bound on the dimension of the  encoding space in  the zero-error protocol 
 described in the proof of Theorem \ref{thm2}.     Precisely, we have the following  
 \begin{lem}\label{lem:upperbounddimensiond}
 The total dimension of all the representation spaces corresponding to Young diagrams with no more than $r$ rows is lower bounded as  
\begin{align}\label{lowerboundsum} \sum_{\lambda\in\mathcal{{Y}}_{N,r}}     \,  d_\lambda  \ge    c (r,d)  \,  N^{\frac{2dr-r^2+ r-2}2} \, , 
\end{align}
where $c$ is a suitable function.        
 \end{lem}

 \Proof  
  For simplicity, we use the notation $f  (N,  r,d)\gtrsim  g(N,r,d)$ to mean that there exists a function  $c (r,d)$  such that $ f(  N,r,d)  \ge c(r,d)   g(N,r,d) $ for every $N$.  If $f(N,r,d)  \gtrsim  g(N,r,d)$  and $g(N,r,d)   \gtrsim  f(N,r,d)$, then we write $  f(N,r,d)  \approx  g(N,r,d)$.    
 With this notation, we have  \begin{align*}
 d_\lambda    \gtrsim    \prod_{1\le i<j\le d}   \,  (\lambda_i -\lambda_j)  \, , 
 \end{align*}
 having used Eq. (\ref{dimensione}).    
 Consider the case when  $N$  is a  multiple of $r(r+1)/2$ and define $s  =   2N/r(r+1)$.    Define the subset of Yang diagrams 
  \begin{align*}     
  \set S_{\rm core}  = \left\{  \lambda  \in  \mathcal Y_{N,r}  ~|~  \lambda_i   \in   \left[ (  r-i+1)  s  -  \frac s {2r} ,  (  r-i+1)  s  +  \frac s{2r}   \right]    \,  ,  \qquad \forall i   =  1,\dots,  r-1  \right\}
 \end{align*} 
For every diagram in $\set S_{\rm core}$ we have the lower bound  
 \begin{align}
 \nonumber d_\lambda     & \gtrsim     \left[    \prod_{1\le i<j\le r-1}   \,  (\lambda_i -\lambda_j)  \right]  \, \left[    \prod_{1\le i\le r-1}   \,  (\lambda_i -\lambda_r)  \right]  \,  \left[        \prod_{1\le i  <r<  j\le d}   \,  \lambda_i   \right]   \,  \left[        \prod_{r<  j\le d}   \,  \lambda_r   \right]  \\
 \nonumber   &  \ge        \left\{    \prod_{1\le i<j\le r}   \,   \left[ ( j-i)s  -  \frac sr   \right]    \right\}  \,  \left\{        \prod_{1\le i\le r-1}   \,   \left[ (  r-i) s  -  \frac s2\right]    \right\}  \,  \left\{        \prod_{1\le i\le r<  j\le d}   \,  (r-i)s   \right\}   \,    \left\{        \prod_{r<  j\le d}   \,     \frac s 2      \right\}   \\
 \nonumber   &  \approx  s^{\frac {  2dr-r^2  -r}2} \\
   \label{core}   & \approx   N^{\frac {  2dr-r^2  -r}2} \, .  
 \end{align}  
Now, the total dimension of the subspaces with Young diagrams in $\set S_{\rm core}$ an be lower bounded as  
\begin{align*}
\sum_{\lambda \in  \set S_{\rm core}}  d_\lambda &  \gtrsim    N^{\frac {  2dr-r^2  -r}2}   \,    |   \set S_{\rm core} |  \\  
&  =   N^{\frac {  2dr-r^2  -r}2}   \,      \left(   \frac sr\right)^{r-1}   \\
&\approx      N^{\frac {  2dr-r^2  -r}2}   \,        N^{r-1}   \\
&  =   N^{  \frac{  2rd-r^2  +r -2}2}  \, .
\end{align*} 
Since $\set S_{\rm core}$ is a subset of $\spc Y_{N,r}$, we obtain Eq. (\ref{lowerboundsum}).  \qed  

\medskip 
Following the steps adopted in the $d=2$ case, it is also  possible to show that the upper bound of Lemma  \ref{lem:upperbounddimensiond}   is actually an upper bound for \emph{every} zero-error protocol  that works for a \emph{complete} ensemble of mixed states---i.~e.~for an ensemble of the form $\{  \rho_g^{\otimes N} \,   ,   p_g \}$ where the state  $\rho_g$ is non-degenerate  and the probability distribution $p_g$ is dense on $\grp {SU}  (d)$.   Essentially, the argument is based on the use of Proposition \ref{prop:todos}, which can be applied here to all the $\grp {SU} (2)$ subgroups of $\grp {SU}  (d)$.

\section{APPROXIMATE COMPRESSION FOR  QUANTUM SYSTEMS OF DIMENSION  $d>2$}

\subsection{Compression protocol}  
Here we consider ensembles of  $N$ identically prepared mixed states, each of them  having the same  spectrum.   Every such ensemble can be written in the form $\{ \rho_g^{\otimes N} ,  p_g  \}$, where $\rho_g$ is a density matrix of the form  
\[\rho_g  =      U_g  \rho_0  U_g^\dag  \, ,  \qquad g\in\grp{SU} (d )  \, ,  \]
$\rho_0$ is a rank-$r$ density matrix with non-degenerate positive eigenvalues, and  $p_g$ is a probability distribution over the group  $\grp {SU}  (d)$.  For ensembles of this form, we have the following

\begin{theo}\label{thm2second}
 For every $\epsilon >0$ there exists an integer $N_0$ such that  for every $N  \ge N_0$ the ensemble  $\{\rho_g^{\otimes N} \, ,  p_g\}$ can be compressed with error less than $\epsilon$ into  $N_{\rm enc}$ qubits, with
 \[N_{\rm enc}  \le  \frac{2dr-r^2-1-m}{2}\log (N+d-1)+    \frac{m+r-1}{2}\log \left[  4 d (d+1)\ln(N+1)+  8    \ln\left(\frac{1}{\epsilon}\right)  +  O\left( \frac 1 {\sqrt N}  \right)\right]    \]
and  $m  :  =  \sum_{i=1}^r  \mu_i$, where  $\mu_i$ be the cardinality of the set $\{ j :  \,  j  >  i  \, ,  p_j  =  p_i \}$. We notice that $m=0$ when the spectrum is non-degenerate.
\end{theo}

The proof of the theorem is based on the Schur-Weyl decomposition 
\begin{align}\label{stateschur} \rho_g^{\otimes N}  =\bigoplus_{\lambda\in\mathcal{{Y}}_{N,r}}q_{\lambda,N} \left(  U_{g}^{(\lambda)}   \,\rho_{0,\lambda}   \,  U_{g}^{(\lambda)  \, \dag}\otimes\frac{I_{m_{\lambda}}}{m_\lambda}\right) \, ,
\end{align}
where    $\rho_{0,\lambda}$ is a fixed density matrix on $\spc R_\lambda$  and $U_g^{(\lambda)}$ is the irreducible representation of $\grp {SU}  (d)$ acting on   $\spc R_\lambda$.  
The  key point is that  the probability distribution  $q_{\lambda, N}$ is concentrated on the Young diagrams such that the vector 
\begin{align}\label{plambda}    p_\lambda   :  =  \left(\frac {\lambda_1} N ,   \dots,  \frac{\lambda_d} N\right)
\end{align}  is close to the vector of the eigenvalues of $\rho_0$ \cite{Spectrum,CM}, listed as 
\begin{align}\label{pspec}  p    =  (p_1, \dots,  p_d)   \,  ,  \qquad p_1\ge p_2  \ge \cdots  \ge p_r  >   p_{r+1}   =  \cdots  =  p_{d}   =  0 \, .  
\end{align}  

Precisely, we will  use the following 
\begin{lem}[\cite{Spectrum,CM}]\label{lem}
Let $p_\lambda$ and $p$ be the vectors defined in Eqs. (\ref{plambda}) and (\ref{pspec}), respectively, and let $ d(a , b)  :  =  \frac 12  \sum_i  |a_i  -b_i|$ be the total variation distance between two vectors. Then, one has
\begin{align*}
{\sf Prob}\left[\lambda: d(p_\lambda, p) >x\right] \le (N + 1)^{d(d+1)/2}\cdot e^{-2Nx^2} \, ,
\end{align*}
with $  {\sf Prob}\left[\lambda: d(p_\lambda, p) >x\right]  :   =   \sum_{\lambda  :  \,   d(p_\lambda, p) >x }  \,    q_{\lambda,N}$, $q_{N, \lambda}$ being the probability distribution in Eq. (\ref{stateschur}).  
\end{lem}

The idea of the proof is to  discard all Young diagrams whose probability vector $p_\lambda$ falls  outside in a ball of size $O(1/\sqrt N)$ around  the vector $p$.      The dimensions of the subspaces associated to the remaining diagrams can be bounded with the following 
\begin{lem}\label{lem:dimensiondegenerate}  
The maximum dimension  of a subspace $\spc R_\lambda$ satisfying   $d (p_\lambda ,  p)  \le  x$ is upper bounded as
\begin{align}
d_\lambda  \le    (4 Nx  + r)^{m}     \,  (  N  + d-1)^{  \frac{2dr-r  (r+1)}2 - m}   \,   . 
\end{align} 
\end{lem}

\Proof  The dimension can be bounded as 
\begin{align*} d_\lambda  & =   \frac{\prod_{1\le i<j\le d}(\lambda_i-\lambda_j-i+j)  }{\prod_{k=1}^{d-1}k!} \\  
&\le   \frac{  \prod_{1\le i\le r}  \left\{   \left[  \prod_{  i< j  \le  i+  \mu_i}   (\lambda_i-\lambda_j-i+j)  \right]  \,  \left[  \prod_{i+ \mu_i<  j  \le d} \,  (\lambda_i-\lambda_j-i+j)\right]  \right\} }{\prod_{k=1}^{d-1}k!} \\
&\le   \frac{  \prod_{1\le i\le r}  \left\{   \left[  \prod_{  i< j  \le  i+  \mu_i}   (4 Nx  + \mu_i)  \right]  \,  \left[  \prod_{i+ \mu_i<  j  \le d} \,  (  N  + d-1)\right]  \right\} }{\prod_{k=1}^{d-1}k!}   \\
& \le   \frac{    
    \prod_{1\le i\le r}   (4 Nx  + \mu_i)^{\mu_i}     \,  (  N  + d-1)^{d-i-\mu_i}    }{\prod_{k=1}^{d-1}k!} \\
    & \le   \frac{    
  (4 Nx  + r)^{m}     \,  (  N  + d-1)^{  \frac{2dr-r  (r+1)}2 - m}    }{\prod_{k=1}^{d-1}k!}   \, ,
\end{align*}   
having used the fact that the ball $\set S  = \left\{\lambda    \in   \spc Y_{N,r}:  \,   d(p_\lambda, p) \le x \right\} $ is contained in the hypercube $\set S'   =   \{  \lambda    \in   \spc Y_{N,r}:    |  \lambda_i/N  -   p_i|  \le 2 x  \, , \forall i   =  1,\dots,  r  -1  \}$, so that, for $p_i  =  p_j$, $i<j$,   one has $\lambda_i  -  \lambda_j \le  4 N x$.     \qed

 \begin{lem}\label{lem:totaldim1}
The total dimension of the subspaces satisfying $d (p_\lambda ,  p)  \le  x$ satisfies 
\[ \sum_{{\lambda    \in   \spc Y_{N,r}:  \,   d(p_\lambda, p) \le x } }  \,    d_\lambda  \le         \,  (  N  + d-1)^{  \frac{2dr-r  (r+1)}2 - m}       (4 Nx  + r)^{m  + r-1} \, .\]
 \end{lem} 

\Proof  
 Immediate from Lemma \ref{lem:dimensiondegenerate}  and from the fact that the ball $\set S  = \left\{\lambda    \in   \spc Y_{N,r}:  \,   d(p_\lambda, p) \le x \right\} $ is contained in the hypercube $\set S'   =   \{  \lambda    \in   \spc Y_{N,r}:    |  \lambda_i/N  -   p_i|  \le 2 x  \, , \forall i   =  1,\dots,  r  -1  \}$, yielding the bound
\[ |\set S|  \le   |\set S'|   \le   (4 Nx)^{r-1}  \, .\]
\qed 

\medskip 

{\bf Proof of Theorem \ref{thm2second}.}  
To compress within an error $\epsilon$, we choose the encoding and decoding channels 
\begin{align*}
\map E     (\rho)     & =    \bigoplus_{\lambda  \in   \set S_\epsilon }      \Tr_{\spc M_\lambda}   [\Pi_\lambda  \rho  \Pi_\lambda]  \,  \oplus  \,  \Tr \left[  \rho   \left(I^{\otimes N}-  \Pi_\epsilon  \right)\right]  \, \rho_{\rm fail}  \\
\map D(\rho')      &    =      \bigoplus_{\lambda  \in \set S_\epsilon} \, \left(       P_\lambda  \rho'   P_\lambda  \otimes \frac {  I_{m_\lambda}}{m_{\lambda}}  \right) \, ,        
\end{align*}
with $\Pi_{\epsilon}  =  \bigoplus_{\lambda\in\set S_\epsilon} \Pi_\lambda$,  $ \Supp ( \rho_{\rm fail})  \subseteq  \spc H_{\rm enc}  = \bigoplus_{\lambda  \in  \set S_\epsilon}  \,    \spc R_\lambda$, and
\begin{align*}
\set S_\epsilon:=\left\{\lambda\in\mathcal{{Y}}_{N,r}~|~ d(p_\lambda, p) \le x_\epsilon\right\} \, , \qquad    x_\epsilon= \sqrt{\frac{d(d+1)/2\ln(N+1)+ \ln(1/\epsilon)}{ 2N}}    \, .
\end{align*}
The value of $x_\epsilon$ is chosen in order to bound the compression error as 
 \begin{align*}
e_N    &=       \frac 12 \left\|   \map D\circ \map E  \left( \rho_g^{\otimes N} \right)   -   \rho_g^{\otimes N}   \right\|     \qquad \forall g\in\grp{SU}  (d)   \\  &\le  \frac 12    \Tr \left[  \rho   \left(I^{\otimes N}-  \Pi_\epsilon  \right)\right]  \,      \left\| \map D  (\rho_{\rm fail})   -       \rho_{g,\rm fail}      \right\|    \, , \qquad \rho_{g, \rm fail} :  =    \bigoplus_{ \lambda\not \in \set S_\epsilon}     \frac{  q_{\lambda,N} }{ \Tr \left[  \rho   \left(I^{\otimes N}-  \Pi_\epsilon  \right)\right]  }  \, \left(  U_{g}^{(\lambda)}   \,\rho_{0,\lambda}   \,  U_{g}^{(\lambda)     \, \dag}    \otimes\frac{I_{m_{\lambda}}}{m_\lambda}    \right)  \\
 &\le    \Tr \left[  \rho   \left(I^{\otimes N}-  \Pi_\epsilon  \right)\right]   \\
 &  =   \,   \sum_{\lambda \not \in  \set S_\epsilon}  \, q_{\lambda,N}  \\
 &\le    (N + 1)^{d(d+1)/2}\cdot e^{-2Nx^2}  \\
 &  =  \epsilon \, , 
\end{align*}
the last inequality coming from Lemma \ref{lem}.  On the other hand, the encoding subspace has dimension
\begin{align*}
d_{\rm enc}&  =  \sum_{\lambda  \in  \set S_
\epsilon}  \,    d_\lambda  \\
  &   \le     (  N  + d-1)^{dr   -  \frac{r  (r+1)}2 - m}       (4 Nx  + r)^{m  + r-1}    \\
  &\le       (  N  + d-1)^{dr   -  \frac{r  (r+1)}2 - m}           \,       N^{\frac{m+ r-1}2}  \,  \left[4d(d+1)\ln (N+1) +8\ln \left(\frac 1\epsilon\right)  +  O\left( \frac 1 {\sqrt N} \right )\right]^{\frac{m+r-1}2}   \\
    &\le   (N+d-1)^{ \frac{2 dr  -  r^2  - 1 - m }2 }       \,  \left[4d(d+1)\ln (N+1) +8\ln \left(\frac 1\epsilon\right)  +  O\left(  \frac1 {\sqrt N} \right )\right]^{\frac{m+r-1}2}    \, , 
\end{align*}
having used Lemma \ref{lem:totaldim1} and the definition of $x_\epsilon$.  Hence,  the number of encoding qubits satisfies 
\begin{align*}
N_{\rm enc}     &\le  \log d_{\rm enc}  \\
 &\le\frac{2rd-r^2-1-m}{2}\log (N+d-1)+    \frac{m+r-1}{2}\log \left[  4 d (d+1)\ln(N+1)+  8    \ln\left(\frac{1}{\epsilon}\right)  +  O\left( \frac 1 {\sqrt N}  \right)\right]   \, .
\end{align*}
\qed 

\subsection{Optimality proof in the presence of symmetry}  
Here we prove the converse of Theorem  \ref{thm2second}.   Our proof is valid for protocols where the encoding is covariant and the decoding preserves the \emph{nonabelian charges} \cite{algebra-preserve}  identified by the Young diagrams.
 Precisely, we assume that 
  \begin{enumerate}
 \item the encoding space $\spc H_{\rm enc}$  supports a unitary representation of the group $\grp {SU}  (d)$, denoted by  $  \{  V_g ~|~ g\in\grp {SU} (d) \}$. 
 \item  the encoding channel satisfies the covariance condition    $ \map E \circ    \map U_g      =   \map  V_g  \circ  \map E$,    $\forall  g\in\grp {SU}  (d)$. 
  \item  the decoding channel $\map D$   preserves the nonabelian charges associated to $\grp {SU}  (d)$, namely,  for every input state $\rho$, one has 
  \begin{align}\label{conservation}
   \Tr  \left [   {\Pi}_\lambda  \,    \map D( \rho)  \right]  =   \Tr  \left[    \widetilde \Pi_\lambda     \,  \rho  \right ]  \,  \qquad \forall \lambda \in  \spc Y_{N,d}       \, ,
 \end{align}
where    $\widetilde \Pi_{\lambda}$   is the projector on the direct sum  of  all the invariant subspaces of $\spc H_{\rm enc}$  with Young diagram $\lambda$.
 \end{enumerate}
By the same argument as in the qubit case, the error of the compression protocol satisfying the above assumption can be lower bounded as $e_N\ge (1/2)\sum_{\lambda\in\set S}q_{\lambda,N}$, with $\set S$ being a subset of $\spc{Y}_{N,r}$ specified by the protocol. The encoding dimension is given by $d_{\rm enc}=\sum_{\lambda\in\set S}q_{\lambda,N}$. We have the following theorem.

\begin{theo}\label{thmopt}     Every compression protocol that encodes a complete  $N$-qubit ensemble  into
$$\left(\frac{2dr-r^2-1-m}{2}-\delta\right)\log N \, ,  \qquad \delta   >  0 \, ,$$
  qubits  with  covariant encoding and a decoding that  preserves the  nonabelian charges  will necessarily have error $e \ge1/2$ in the asymptotic limit. Here  $m  :  =  \sum_{i=1}^r  \mu_i$, where  $\mu_i$ be the cardinality of the set $\{ j :  \,  j  >  i  \, ,  p_j  =  p_i \}$. We notice that $m=0$ when the spectrum is non-degenerate.
\end{theo}

To prove the theorem, we first define the cubic lattice
\begin{align}\label{H-epsilon}
\set H_\epsilon   =   \left\{  \lambda  \in  \spc Y_{N,r}   ~\left|~  
\lambda_i  \in  \left[  p_i N -    \frac{\sqrt{  N ^{1+\epsilon}}} 2 ,   p_i N +  \frac{\sqrt{N^{1+\epsilon}}} 2   \right]     \, ,  \quad  \forall~i=1,\dots,  r-1 
\right\}\right.  \, 
\end{align}
for any constant $\epsilon\in(0,1)$. With this definition, the sum of the probability $q_{\lambda,N}$ when $\lambda\not\in\set H_\epsilon$ vanishes exponentially in $N$. Precisely, we have the following lemma.
\begin{lem} \label{lemma-H-epsilon}
For the set $\set H_\epsilon$ defined by Eq. (\ref{H-epsilon}), the following bound holds.
$$\sum_{\lambda\not\in \set H_{\epsilon}}q_{\lambda,N}\le (N+1)^{\frac{d(d+1)}2}e^{-\frac{N^\epsilon}8}.$$
\end{lem}
\begin{proof}
For any Young diagram $\lambda$ not in the set $\set H_\epsilon$, there exist at least one $j$ such that $|\lambda_j-p_i N|\ge  \sqrt{  N ^{1+\epsilon}}/2$.
Thus we have $$d(p_\lambda,p)\ge \frac12\left|\frac{\lambda_j}{N}-p_j\right|\ge \frac1{4\sqrt{N^{1-\epsilon}}}.$$ Substituting this fact into Lemma \ref{lem}, we immediately get the following lemma.
\end{proof}

Now we start to bound the probability distribution $q_{\lambda,N}$ within the set $\set H_\epsilon$. Notice that the exact expression of $q_{\lambda,N}$ is given as \cite{ExactDistribution}
\begin{align}\label{q-expression}
q_{\lambda,N}=\frac{\det \Delta }{\det \Sigma}\cdot m_{\lambda}
\end{align}
where the matrix $\Sigma$ is independent of $N$ (and thus its expression is not relevant to bounding the probability) and the matrix $\Delta$ is a rank $r$ square matrix defined as the following.
\begin{align}
\Delta_{ij}=\left[\prod_{\beta=0}^{\mu_j-1}(\lambda_i+r-i-\beta)\right]p_j^{\lambda_i+r-i-\mu_j},
\end{align}
with $\mu_i$ defined in Theorem \ref{thmopt}. Notice that we follow the convention $\prod_{i=0}^{-1}f(i)=1$. We first prove the following bound of $\det\Delta$.

\begin{lem} \label{lemma-delta}
For any $\lambda$ in the set $\set H_\epsilon$ defined by Eq. (\ref{H-epsilon}), the following bound holds asymptotically for large $N$: 
$$\det\Delta\lesssim N^{\frac {(1+\epsilon)m}2}\left(\prod_{i=1}^{r}p_i^{\lambda_i}\right) \, , \qquad m=\sum_{i=1}^r \mu_i \, . $$

\end{lem}

\begin{proof}
Suppose that there are $k$ distinct positive values in the spectrum, and the $i$-th biggest value has degeneracy $r_i$. 
We can then divide the set $\{1,\dots,r\}$ into $k$ subsets $\set L_1\cup\dots\cup\set L_k$, corresponding to the distinct eigenvalues, so that  $\set L_i$ is the set of  indices corresponding to the $i$-th biggest eigenvalue.  Recalling that $r_j$ is the degeneracy of the $j$-th eigenvalue, we have
$$\set L_i=\left\{\sum_{j=1}^{i-1}r_j+1,\dots,\sum_{j=1}^i r_j\right\}.$$
Notice that, by definition, one has 
\begin{align}\label{L-property}
p_l=p_k\qquad\forall\,l,k\in\set L_i.
\end{align}
With the above definition, the spectrum now reads
\begin{align*}
\underbrace{p_1=\dots=p_{r_1}}_{\set L_1}>\underbrace{p_{r_1+1}=\dots=p_{r_1+r_2}}_{\set L_2}>\dots>\underbrace{p_{\sum_{i=1}^{k-1}r_i+1}=\dots=p_{r}}_{\set L_k}>p_{r+1}=\dots=p_{d}=0.
\end{align*}

Correspondingly, we define a subgroup $\grp{P}_r$ of the group $\grp S_r$, consisting of the product of permutations that act within the subsets  $\{\set L_i\}$. Precisely, 
\begin{align*}
\grp{P}_r:=\left\{  \sigma^{(1)}  \times \sigma^{(2)}  \times \cdots \times \sigma^{(k)} \,  | \,  \sigma^{(i)} \in\grp S_{r_i}; i=1,\dots,k\right\}.
\end{align*}

 With the above definition, we divide $\det\Delta$ into two terms
\begin{equation}\label{main-bound}
\begin{split}
\det \Delta &=t_1+t_2\\
 t_1&=\sum_{\sigma\in\grp{P}_r}{\rm sgn}(\sigma)\left(\prod_{i=1}^{r}\Delta_{i\,\sigma_i }\right)\\
 t_2&=\sum_{\sigma\not\in\grp{P}_r}{\rm sgn}(\sigma)\left(\prod_{i=1}^{r}\Delta_{i\,\sigma_i }\right),
 \end{split}
\end{equation}
denoting by $\sigma_i$ the index that comes from applying $\sigma$ to $i$.

Let us bound $t_1$.   By definition, $\grp P_r$ contains every permutation $\sigma$ such that $p_i=p_{\sigma_i}$ for every $i$. Therefore, we have
\begin{align*}
 t_1&=  \sum_{\sigma\in\grp{P}_r}{\rm sgn}(\sigma)\left\{\prod_{i=1}^{r}\left[\prod_{\beta=0}^{\mu_{\sigma_i}-1}(\lambda_i+r-i-\beta)\right]p_{\sigma_i}^{\lambda_i+r-i-\mu_{\sigma_i}}\right\}\\
   & =  \sum_{\sigma\in\grp{P}_r}{\rm sgn}(\sigma)\left\{\prod_{i=1}^{r}\left[\prod_{\beta=0}^{\mu_{\sigma_i}-1}(\lambda_i+r-i-\beta)\right]p_{i}^{\lambda_i+r-i}\right\}\left(\prod_{i=1}^{r}p_{\sigma_i}^{-\mu_{\sigma_i}}\right)\\
      & =  \sum_{\sigma\in\grp{P}_r}{\rm sgn}(\sigma)\left\{\prod_{i=1}^{r}\left[\prod_{\beta=0}^{\mu_{\sigma_i}-1}(\lambda_i+r-i-\beta)\right]p_{i}^{\lambda_i+r-i}\right\}\left(\prod_{i=1}^{r}p_{i}^{-\mu_i}\right)\\
 & =  \left(\prod_{i=1}^{r} p_{i}^{\lambda_i+r-i-\mu_{i}}\right)\sum_{\sigma\in\grp{P}_r}{\rm sgn}(\sigma)\left[\prod_{i=1}^{r}\prod_{\beta=0}^{\mu_{\sigma_i}-1}(\lambda_i+r-i-\beta)\right].
\end{align*}
Since $i$ and $\sigma_i$ are always in the same subset $\set L_l$ (for  suitable $l$), we can rewrite the term $\prod_{i=1}^{r}\prod_{\beta=0}^{\mu_{\sigma_i}-1}(\lambda_i+r-i-\beta)$ as $\prod_{l=1}^{k}\prod_{i\in \set L_l}\prod_{\beta=0}^{\mu_{\sigma_i}-1}(\lambda_i+r-i-\beta)$. We then have
\begin{align*}
t_1  & =  \left(\prod_{i=1}^{r} p_{i}^{\lambda_i+r-i-\mu_{i}}\right) \sum_{\sigma\in\grp{P}_r}{\rm sgn}(\sigma)\left\{\prod_{l=1}^{k}\left[\prod_{i\in \set L_l}\prod_{\beta=0}^{\mu_{\sigma_i}-1}(\lambda_i+r-i-\beta)\right]\right\}\\
    &=\left(\prod_{i=1}^{r} p_{i}^{\lambda_i+r-i-\mu_{i}}\right) \prod_{l=1}^{k}\left\{\sum_{\sigma^{(l)}\in\grp{S}_{r_l}}{\rm sgn}\left(\sigma^{(l)}\right)\left[\prod_{i\in \set L_l}\prod_{\beta=0}^{\mu_{\sigma^{(l)}_i}-1}(\lambda_i+r-i-\beta)\right]\right\}\\
 & = \left(\prod_{i=1}^{r} p_{i}^{\lambda_i+r-i-\mu_{i}}\right) \prod_{l=1}^{k}\left\{\sum_{\sigma^{(l)}\in\grp{S}_{r_l}}{\rm sgn}\left(\sigma^{(l)}\right)\left[\prod_{i\in\set L_l}\left(\Delta_l\right)_{i\sigma^{(l)}_{i}}\right]\right\}\\
  & =  \left(\prod_{i=1}^{r} p_{i}^{\lambda_i+r-i-\mu_{i}}\right)\left(\prod_{l=1}^{k}~\det \Delta_l\right).
\end{align*}
Here $\Delta_l$ is a rank $r_l$ square matrix defined as
\begin{align*}
\left(\Delta_l\right)_{ij}=\prod_{\beta=0}^{r_l-j-1}(\lambda_i+r-i-\beta),
\end{align*}
observing that $\mu_j$ assumes the values $r_l-1,r_l-2,\dots,1,0$ for the indices in $\set L_l$.
The determinant of $\Delta_l$ equals to $\prod_{1\le i<j\le r_l}(\lambda_i-\lambda_j+j-i)$. Combining this with the definition of $\set H_\epsilon$ (\ref{H-epsilon}), we have
\begin{align}
t_1&=\left[\prod_{l=1}^{k}\prod_{1\le i<j\le r_l}(\lambda_i-\lambda_j+j-i)\right]\left(\prod_{i=1}^{r} p_{i}^{\lambda_i+r-i-\mu_{i}}\right)\nonumber\\
&\lesssim \left[\prod_{l=1}^{k}\left(\sqrt{N^{1+\epsilon}}\right)^{\frac{r_l(r_l-1)}2}\right]\left(\prod_{i=1}^{r} p_{i}^{\lambda_i+r-i-\mu_{i}}\right)\nonumber\\
&= N^{\frac{(1+\epsilon)m}{2}}\left(\prod_{i=1}^{r} p_{i}^{\lambda_i+r-i-\mu_{i}}\right)\nonumber\\
&\approx N^{\frac{(1+\epsilon)m}{2}}\left(\prod_{i=1}^{r} p_{i}^{\lambda_i}\right)\label{term1}.
\end{align}
The last step follows from the fact that $$m=\sum_{i=1}^{r} \mu_i=\sum_{i=1}^{k}\sum_{j=1}^{r_i}(r_i-j).$$

Next, we bound the second term $t_2$ in Eq. (\ref{main-bound}) as
\begin{align*}
t_2&\le \sum_{\sigma\not\in\grp{P}_r}\left(\prod_{i=1}^{r}\Delta_{i\,\sigma_i }\right)\\
&= \sum_{\sigma\not\in\grp{P}_r}\left\{\prod_{i=1}^{r}\left[\prod_{j=0}^{\mu_{\sigma_i}-1}(\lambda_i+r-i-j)\right]p_{\sigma_i}^{\lambda_i+r-i-\mu_{\sigma_i}}\right\}\\
&\le \sum_{\sigma\not\in\grp{P}_r}\left[\prod_{i=1}^{r}(N+r-1)^{\mu_{\sigma_i}}p_{\sigma_i}^{\lambda_i+r-i-\mu_{\sigma_i}}\right]\\
&= (N+r-1)^m\sum_{\sigma\not\in\grp{P}_r}\left[\prod_{i=1}^{r}p_{\sigma_i}^{\lambda_i+r-i-\mu_{\sigma_i}}\right]\\
&= (N+r-1)^m\sum_{\sigma\not\in\grp{P}_r}\left[\prod_{i=1}^{r}\left(\frac{p_{\sigma_i}}{p_{i}}\right)^{\lambda_i+r-j-\mu_{\sigma_j}}\right]\left[\prod_{j=1}^{r}p_{j}^{\lambda_j+r-j-\mu_{\sigma_j}}\right]\\
&= (N+r-1)^m\sum_{\sigma\not\in\grp{P}_r}\left[\prod_{i=1}^{r}\left(\frac{p_{\sigma_i}}{p_{i}}\right)^{Np_i+O(\sqrt{N^{1+\epsilon}})}\right]\left[\prod_{j=1}^{r}p_{j}^{\lambda_j+r-j-\mu_{\sigma_j}}\right]\\
&\approx (N+r-1)^m\sum_{\sigma\not\in\grp{P}_r}\exp\left[-N D(p||\sigma_p)\right]\left[\prod_{i=1}^{r}p_{i}^{\lambda_i+r-i-\mu_{\sigma_i}}\right],
\end{align*}
where $D(p||q):=\sum_i p_i \ln(p_i/q_i)$ is the Kullback-Leibler divergence and $\sigma_p:=(\sigma_{p_1},\dots,\sigma_{p_r})$. Now, since $\sigma\not\in\grp{P}_r$, we always have $D(p||\sigma_p)>0$. Therefore, the second term in Eq. (\ref{main-bound}) vanishes exponential in $N$. Combining this fact with Eq. (\ref{main-bound}) and Eq. (\ref{term1}) we get the desired bound on $\det \Delta$.
\end{proof}

\begin{lem} \label{lemma-q/d}
For any $\lambda$ in the set $\set H_\epsilon$ defined by Eq. (\ref{H-epsilon}), the following bound holds asymptotically for large $N$.
$$\frac{q_{\lambda,N}}{d_{\lambda}}\lesssim N^{-\frac{2dr-r^2-1-(1+\epsilon)m}2}.$$
\end{lem}

\begin{proof}
The dimension of $\spc M_\lambda$ is given by   
\begin{align*}
m_\lambda&=\frac{N!}{\prod_{i=1}^{d}(\lambda_i+d-i)!}\prod_{1\le i<j\le d}(\lambda_i-\lambda_j+j-i)
\end{align*}
(see e.~g.~\cite{ExactDistribution}) and can be bounded  as
\begin{align*}
m_\lambda &\le     \frac 1  {  \lambda_1^{d-1}  \,  \lambda_2^{d-2}  \,\dots \, \lambda_r^{d-r} }\,  {N  \choose \lambda}   \prod_{1\le i<j\le d}(\lambda_i-\lambda_j+  j-i)   \\
&\lesssim N^{-\frac{2dr-r^2-r}2}    \,  {N  \choose \lambda}   \prod_{1\le i<j\le d}(\lambda_i-\lambda_j + j-i)     \, 
\end{align*}
for any $\lambda\in\set H_\epsilon$. Substituting the above bound and the bound in Lemma \ref{lemma-delta} into Eq. (\ref{q-expression}), we have
\begin{align*}
q_{\lambda,N}&\lesssim \frac{N^{\frac {(1+\epsilon)m}2}}{\det \Sigma} \left(\prod_{i=1}^{r}p_i^{\lambda_i}\right)\cdot N^{-\frac{2dr-r^2-r}2}    \,  {N  \choose \lambda}   \prod_{1\le i<j\le d}(\lambda_i-\lambda_j + j-i) \\
&\lesssim N^{-\frac{2dr-r^2-r-(1+\epsilon)m}2} m(N,p,\lambda)   \prod_{1\le i<j\le d}(\lambda_i-\lambda_j + j-i) \\
&\lesssim N^{-\frac{2dr-r^2-1-(1+\epsilon)m}2}  \prod_{1\le i<j\le d}(\lambda_i-\lambda_j + j-i) \\
\end{align*}
which holds for any $\lambda\in\set H_\epsilon$. The last inequality comes from the upper bound of the multinomial $m(N,p,\lambda)$. Finally, we get the desired bound of $q_{\lambda,N}/d_\lambda$ by combining the above bound with the expression of $d_{\lambda}$ 
$$d_\lambda=\frac{\prod_{1\le i<j\le d}(\lambda_i-\lambda_j-i+j)}{\prod_{k=1}^{d-1}k!}.$$
\end{proof}

Finally, we can bound the error of any compression protocol with an encoding set $\set S$ and with the encoding dimension $d_{\rm enc}=O\left(N^{\frac{2dr-r^2-1-m}{2}-\delta}\right)$ as
\begin{align*}
e_N&\ge \frac12\sum_{\lambda\in\set S} q_{\lambda,N}\\
&=\frac12\left(1-\sum_{\lambda\not\in\set S} q_{\lambda,N}\right)\\
&\ge\frac12\left(1-\sum_{\lambda\not\in \set H_{\delta/m}}q_{\lambda,N}-\sum_{\lambda\in\set H_{\delta/m}\cap\set S} q_{\lambda,N}\right)\\
&\ge\frac12\left[1-\sum_{\lambda\not\in \set H_{\delta/m}}q_{\lambda,N}-\max_{\lambda\in\set H_{\delta/m}}\left(\frac{q_{\lambda,N}}{d_{\lambda}}\right)\sum_{\lambda\in\set S}d_{\lambda}\right]\\
&\ge\frac12\left[1-\sum_{\lambda\not\in \set H_{\delta/m}}q_{\lambda,N}-\max_{\lambda\in\set H_{\delta/m}}\left(\frac{q_{\lambda,N}}{d_{\lambda}}\right)\cdot d_{\rm enc}\right]\\
&\gtrsim\frac12\left[1-(N+1)^{\frac{d(d+1)}2}e^{-\frac{1}8 N^{\frac{\delta}m}}-N^{-\frac{\delta}{2}}\right]\\
&=\frac12\left(1-N^{-\frac{\delta}{2}}\right).
\end{align*}

\end{widetext}

\end{document}